\documentclass[superscriptaddress]{revtex4-2}
\usepackage{graphicx} 
\usepackage{wrapfig}
 \usepackage[T1]{fontenc}
\usepackage[utf8]{inputenc}
\usepackage{physics}
\usepackage{bbold}
\usepackage{amsmath,amsthm}
\usepackage{amssymb}
\usepackage[a4paper, total={6in, 8in}]{geometry}
\usepackage[table,xcdraw]{xcolor}
\usepackage{caption}
\usepackage{subcaption}
\usepackage{setspace}
\usepackage{verbatim}
\usepackage[rightcaption]{sidecap}
\usepackage{bm,bbm}
\usepackage{tikz}
\usetikzlibrary{decorations.pathmorphing}
\usepackage{circuitikz}
\usepgflibrary{patterns}
\usetikzlibrary {patterns,patterns.meta} 
\usepackage{mathtools}
\usepackage{nicefrac}

\usepackage{enumitem}
\usepackage{hyperref}

\usepackage{color}

\newcommand{\h}{\mathcal{H}}
\newcommand{\B}{\mathcal{B}}
\newcommand{\s}{\mathcal{S}}
\newcommand{\R}{\mathbb{R}}
\newcommand{\C}{\mathcal{C}}

\newcommand{\bone}{\mathbbm{1}}	
\newcommand{\E}{\mathcal{E}}
\newcommand{\vc}{\vcentcolon =}
\newcommand{\F}{\mathcal{F}}
\renewcommand{\P}{\mathcal{P}}

\newcommand{\dt}{\partial}

\usepackage{color}
\usepackage{xcolor}
\usepackage{framed}

\definecolor{shadecolor}{gray}{0.8}
\definecolor{lgray}{gray}{0.5}

\newcounter{mnotecount}[section]
\renewcommand{\themnotecount}{\thesection.\arabic{mnotecount}}
\newcommand{\mnote}[1]
{\protect{\stepcounter{mnotecount}}$^{\mbox{\footnotesize
$
\bullet$\themnotecount}}$ \marginpar{
\raggedright\tiny\em
$\!\!\!\!\!\!\,\bullet$\themnotecount: #1} }

\begin{document}

\newtheorem{prop}{Proposition}
\newtheorem{theorem}{Theorem}
\newtheorem{corollary}{Corollary}[theorem]
\newtheorem{lemma}[theorem]{Lemma}
\newtheorem{remark}{Remark}

\title{Superluminal signalling and chaos in nonlinear quantum dynamics}

\author{Marta Emilia Bieli\'nska}
\email{marta.bielinska@oriel.ox.ac.uk}
\affiliation{Oriel College, University of Oxford, 
Oriel Square, Oxford OX1 4EW, UK}
\affiliation{International  Centre  for  Theory  of  Quantum  Technologies,  University  of  Gda\'nsk,  Wita  Stwosza  63,  80-308  Gda\'nsk,  Poland}

\author{Micha\l\ Eckstein}
\email{michal.eckstein@uj.edu.pl}
\affiliation{Institute of Theoretical Physics, Jagiellonian University, {\L}ojasiewicza 11, 30-348 Krak\'ow, Poland}
\affiliation{International  Centre  for  Theory  of  Quantum  Technologies,  University  of  Gda\'nsk,  Wita  Stwosza  63,  80-308  Gda\'nsk,  Poland}

\author{Pawe\l\ Horodecki}
\affiliation{International  Centre  for  Theory  of  Quantum  Technologies,  University  of  Gda\'nsk,  Wita  Stwosza  63,  80-308  Gda\'nsk,  Poland}
\affiliation{Faculty of Applied Physics and Mathematics, National Quantum Information Centre,
Gda\'nsk University of Technology, Gabriela Narutowicza 11/12, 80-233 Gda\'nsk, Poland}

\date{\today}

\maketitle 

\spacing{1.2}

\section*{Abstract}
Nonlinear quantum dynamics is often invoked in models trying to bridge the gap between the quantum micro-world and the classical macro-world. Such endeavors, however, encounter challenges at the nexus with relativity. In 1989 Nicolas Gisin proved a powerful no-go theorem, according to which nonlinear quantum dynamics would lead to superluminal signalling, violating Einstein’s causality.
Here we analyse the theorem from the perspective of recent developments. First, we observe that it harmonises with the no-restriction hypothesis from General Probabilistic Theories. Second, we note that it requires a suitable synchronisation of Alice's and Bob's clocks and actions. Next, we argue that it does not automatically exclude the possibility of global nonlinear quantum dynamics on a tensor product Hilbert space. Consequently, we investigate a class of such dynamics inspired by discrete analogues of nonlinear Schr\"odinger equations. We show that, in general, they exhibit a chaotic character. In this context we inspect whether superluminal signalling can be avoided by relaxing the no-restriction hypothesis. We study three possible communication protocols involving either local measurements or modifications of a local Hamiltonian. We conclude that, in general, in all three cases, two spacelike separated parties can effectuate statistical superluminal information transfer. Nevertheless, we show an example of a nonlocal nonlinear quantum dynamics, which does not allow for it, provided that we relax the no-restriction hypothesis.

\newpage
\tableofcontents

\newpage
\section{Introduction}

Quantum mechanics has a distinctive mathematical feature --- it is a \textit{linear} theory. This contrasts with classical mechanics, which is based on nonlinear structures within the domain of differential geometry. The linearity of quantum mechanics is manifest at both the kinematic and dynamical levels. The former pertains to the fact that the spaces of quantum states and observables are linear topological spaces. The latter means that the fundamental evolution equations --- Schr\"odinger, Heisenberg, von Neumann or Gorini--Kossakowski--Sudarshan--Lindblad \cite{GKS,Lindblad} --- are linear partial differential equations. 

The linearity of the quantum formalism is based on the superposition principle, which has been questioned from different standpoints. A programme for establishing a causal nonlinear wave mechanics has already been outlined by de Broglie as a continuation of his pilot-wave formalism \cite{deBroglie1960}. The first concrete equation involving a logarithmic nonlinearity was proposed by Białynicki-Birula and Mycielski \cite{IBB_nonlinear}. Another version of `nonlinear quantum mechanics' was put forward by Weinberg \cite{WeinbergNQM,WeinbergNQM2}. The common motivation behind these models was to describe a mesoscopic middle ground between the quantum micro-world and the classical macro-world. Similar ideas underlie the models of Karolyhazy \cite{Karolyhazy}, Di'osi \cite{Diosi87}, and Penrose \cite{Penrose96}, which postulate the mediation of gravity in the quantum-to-classical transition. These models belong to a wide class of wave function collapse models \cite{OR}, which are being tested in modern experiments \cite{Superpos2019,Diosi21,collapse22}.

The central difficulty with nonlinear quantum dynamics is that it appears to be incompatible with relativity. More precisely, Simon, Bu\v{z}ek, and Gisin have proven \cite{Gisin1989,Gisin2001} that the joint assumptions of the static structure of quantum mechanics and the absence of superluminal signaling imply that the dynamics of quantum states must be described by a completely positive, trace-preserving linear map.

A way around this problem, commonly adopted in modern wave function collapse models \cite{OR}, is to introduce some stochastic element, which results in linear dynamics of density matrices and thus prevents causality violation \cite{Bassi_NS}. A universal prediction of these models is the existence of a ``collapse noise'', which leads to the heating of the system \cite{collapse22}. This predicted effect allowed for setting rather stringent bounds on various stochastic collapse models \cite{Diosi21,collapse22}. 

There also exist several specific models of deterministic nonlinear quantum dynamics that circumvent the problem of superluminal signalling. Typically, such models introduce some modification of the collapse postulate and/or the Born rule \cite{Czachor1998,CzachorMarcin,Czachor02,Kent05,Helou17}. More recently, Rembieli\'nski and Caban \cite{Caban20,Caban21} have shown that Gisin's argument does not rule out a class of so-called `convex quasilinear' dynamical maps, which include, i.a., the selective measurement map. Another deterministic nonlinear variant of quantum dynamics was proposed in \cite{Kaplan2022}. The latter framework involves a modification at the level of quantum field theory and postulates that no collapse actually takes place.

\medskip

The goals of this paper are twofold:

First, we aim to elucidate the no-go theorem on nonlinear quantum dynamics presented in \cite{Gisin2001} by carefully identifying all assumptions underlying its derivation. In particular, the theorem presumes that the signalling party (Alice) is free to choose any local POVM and that the receiver (Bob) can reconstruct his local state with arbitrary precision. In the context of General Probabilistic Theories \cite{GPT_review}, these assumptions fall under the `no-restriction hypothesis' \cite{Chiribella10,No_restrictions}. We also note that while Gisin's argument does not require the von Neumann state-update postulate, it does employ a rather non-trivial assumption about the spacetime aspect of the collapse.

We further argue that the no-go theorem from \cite{Gisin2001} assumes that the studied dynamical map is `local' (cf. \cite{Bell_Nonlocal}), in the sense that it acts exclusively on the Hilbert space associated with the subsystem of Alice or Bob. Moreover, the map is defined only on pure states and there is no natural extension to mixed states. This fact was exploited in \cite{Caban20,Caban21} where the authors provided an explicit example of a nonlinear, yet non-signalling, local dynamics of \textit{mixed} quantum states. 

Secondly, we aim to explore the possibility of a `nonlocal' nonlinear quantum dynamics that does not facilitate operational superluminal signalling. To this end, we analyse a class of dynamical equations on a finite-dimensional bipartite Hilbert space, inspired by some popular effective nonlinear Schr\"odinger equations, such as the Gross–-Pitaevskii equation, which is widely applied in Bose–-Einstein condensate \cite{BEC_review}. Such finite-dimensional nonlinear evolution equations for single systems were studied in \cite{MeyerWong13,MeyerWong14,Childs} in the context of quantum information processing. In particular, in \cite{Childs} it was demonstrated that nonlinear dynamics of a single qubit leads to an exponential increase in distinguishability between two states with a large initial overlap. Here, we illustrate the same phenomenon in a system of two coupled qubits. We also show that the dynamics of entanglement is not monotonic and exhibits a chaotic character.

Then, we consider three protocols that Alice and Bob could employ to establish superluminal communication with the help of nonlocal nonlinear quantum dynamics. The first one involves a choice of measurement basis, as in the protocol employed in Gisin's original argument \cite{Gisin1989}. In the second one, Alice can only make a measurement in a single distinguished basis, but she is free to decide whether to make the measurement or not. In the third protocol Alice exploits the possibility of freely modifying the local dynamics of her particle (i.e., quantum-information-carrier), without making any measurement. We find that in all three cases superluminal communication is, in general, possible. Furthermore, we show that, due to the chaotic character of the dynamics, a tiny change on Alice's side typically quickly develops into a sizeable effect on Bob's side.

Nevertheless, we demonstrate that there exist instances of nonlocal nonlinear quantum dynamics which does allow for signalling, neither through projective measurements nor through local modification of dynamics. This comes at the cost of a radical violation of the no-restriction hypothesis --- the measurements can only be performed in the basis distinguished by the dynamics.

\medskip

The article is organised as follows: In Section \ref{sec:2}, we scrutinise Gisin's no-go theorem based on \cite{Gisin2001}. We argue that its applicability is limited by the notorious problem of extending dynamical maps to composite systems. In Section \ref{sec:3}, we present a general class of norm-preserving nonlinear dynamics on a finite-dimensional Hilbert space. Such dynamics is always nonlocal and introduces a preferred basis. We consider a bipartite scenario with spacelike separated Alice and Bob and analyse the spacetime diagram associated with it. We also present some analytic solutions to the dynamical equations for a special case involving diagonal Hamiltonians. Then, in Section \ref{sec:4}, we focus on the simplest two-qubit system driven by a discrete analogue of the Gross–-Pitaevskii equation. We illustrate the general chaotic character of the dynamics. We do so by studying the time evolution of the overlap between two states, the entanglement quantified by the concurrence, and the dynamics of Bob's reduced state in the Bloch ball. Section \ref{sec:sign} contains the analysis of three different superluminal communication protocols. In the first two, we assume the standard von Neumann collapse postulate, while in the last, Alice encodes information in the system by modifying the parameters of a local Hamiltonian.

Finally, in Section \ref{sec:no}, we discuss special cases of nonlocal nonlinear quantum dynamics that do not facilitate superluminal signalling through local projective measurements and modifications of dynamics. They are rather heavily restricted by the requirement that both the Hamiltonian part and the measurements be diagonal in the preferred basis, thus sharply violating the no-restriction hypothesis. We conclude in Section \ref{sec:conc} with a short summary and an outlook on the possible use of physical models based on nonlinear quantum dynamics.

\section{Gisin's theorem}
\label{sec:2}

In 1989, Nicolas Gisin \cite{Gisin1989} proved a powerful no-go theorem, which implies that nonlinear modifications of quantum dynamics typically lead to superluminal signalling (cf. also \cite{Bassi_NS}).

Let us briefly recall the theorem in a refined version presented in \cite{Gisin2001}, co-authored by Simon,  Bu\v{z}ek and Gisin. Let $\h = \h_A \otimes \h_B$ be a bipartite Hilbert space and let $S(\h)$ denote the space of mixed states, i.e.  density operators on $\h$. Suppose that two spacelike separated parties, Alice and Bob, share a pair of quantum systems (say, `particles') in an entangled state $\ket{\Psi} \in \h$. It is assumed that these agents can interact \emph{locally} with their systems and their local environments (see Fig. \ref{fig:Gisin} \textbf{a)}).  
More specifically, any Alice's operation is implemented as a map $X_{AA'} \otimes \bone_{BB'}$ acting on the Hilbert space $\h_{AA'} \otimes \h_{BB'} \vc \h_A \otimes \h_{A'} \otimes \h_B \otimes \h_{B'}$, where $\h_{A'}$ and $\h_{B'}$ refer to Alice's and Bob's laboratories. 

At a given moment, Alice makes a local (generalised) measurement of an observable on this state. By a suitable choice of her observable, Alice can prepare at a distance any ensemble of local (pure) states, $\{p_i,\ket{\psi_i}\}_i$, with $\ket{\psi_i} \in \h_B$, $0 \leq p_i \leq 1$, $\sum_i p_i = 1$, for Bob. 
Then, Bob inputs his local state into an information-processing device, characterised by a map 
$\E: \h_B \to S(\h_B)$, which results in the statistical combination of final states $\rho_B' = \sum_i p_i\, \E(\ket{\psi_i})$. Eventually, Bob makes a quantum measurement on the effective output state $\rho_B'$.

The theorem proved in \cite{Gisin2001} shows that if Alice can chose between two local observables, $\{P_A^n\}_{n=1,2}$, and hence prepare at-a-distance two different ensembles, $\{p_i^n,\ket{\psi_i^n}\}_i$, then Bob's final state $\rho_B'$ does not depend on Alice's input, that is 
$\rho_B' = \sum_i p_i^1\, \E(\ket{\psi_i^1}) = \sum_i p_i^2\, \E(\ket{\psi_i^2})$, 
\emph{if and only if} the map $\E$ is linear. 
Putting it the other way around, if $\E$ is not linear then Alice can (statistically) send a bit of information to the spacelike separated Bob.

As emphasised in \cite{Gisin2001}, the argument does not require any collapse postulate but follows directly from the consistency of using the Born rule for local and joint probabilities. On the other hand, it relies on a few assumptions, which we unveil below.

The first assumption is explicitly phrased in \cite{Gisin2001} as follows:
\begin{itemize}
    \item[(1a)] Physical states are described by rays in a Hilbert space $\h$ and observables are described by projectors on $\h$. 
\end{itemize}
What is silently assumed, on top of (1a), is that:
\begin{itemize}
    \item[(1b)] The Hilbert space $\h$ is a tensor product, $\h = \h_{AA'} \otimes \h_{BB'}$.
    \item[(1c)] All measurements accessible to Alice and Bob are \emph{local}, that is they are described in terms of projectors of the form $P_{AA'} \otimes \bone$ and $\bone \otimes P_{BB'}$, respectively.
    \item[(1d)] There exists at least one entangled state $\ket{\Psi} \in \h$, which describes the physical state of the system at hand.
\end{itemize}
Assumptions (1b) and (1c) are the aforementioned mentioned `locality' requirements, which do not follow from postulate (1a). 
Assumption (1d) is, in fact, essential for the argument, as emphasised in \cite{Gisin2001}. Assumption (1a) on its own does not specify which states are physical and it allows for modifications of quantum mechanics, in which all physical states are described by a subspace of product states in $\h_{AA'} \otimes \h_{BB'}$.

The second explicit assumption from \cite{Gisin2001} is the standard trace rule:
\begin{itemize}
	\item[(2)] Probabilities for measurement outcomes at any given time are calculated according to the usual trace rule.
\end{itemize}
This postulate  does not require any further specification, except of the way in which it employs the notion of time, which we discuss below.

Finally, the pivotal assumption which leads to the conclusion that the time-evolution of states must be linear completely positive is the \emph{no-signalling principle}:
\begin{itemize}
	\item[(3a)] Superluminal communication between Alice and Bob is excluded.
\end{itemize}
Obviously, this statement implicitly involves the spacetime structure, or at least its part associated with the causal order of events. The spacetime diagram associated with the would-be superluminal communication protocol from Alice to Bob is presented in Fig. \ref{fig:Gisin} \textbf{a)}. This diagram shows that Alice's measurement and Bob's triggering of the information processing must be synchronised. To that end, first, Alice and Bob must synchronise their local clocks (which can be done with the help of the source
 located in their common causal past). But, more importantly, Alice and Bob must also know \emph{when and where} will the remote pure state preparation take place. 
 This is essential for Gisin's argument, because the dynamics, characterised by the map $\E$, is not \textit{a priori} specified neither for mixed states nor for parts of a global entangled system \footnote{Note also the that remote state preparation requires the global state to be entangled at the moment of Alice's measurement, which in particular means that Bob is not allow to measure the system before Alice and the application of his local dynamics.}. Hence, the argument requires an additional assumption:  
\begin{itemize}
	\item[(3b)] The state preparation at-a-distance on Bob's side, effectuated by Alice's measurement, takes place immediately on a \emph{fixed} spacelike hypersurface. 
\end{itemize}
This is, in fact, a rather strong assumption. First, it says that there exists \emph{some} preferred time-slicing in spacetime and, second, that the local measurement process has a nonlocal effect. 

\begin{figure}[h]
\begin{center}
\hspace*{-2.2cm}
\resizebox{0.54\textwidth}{!}{
\begin{circuitikz}
\tikzstyle{every node}=[font=\LARGE]
\node [] at (4,3) {\Huge{\textbf{a)}}};
\draw [line width=0.8pt, ->, >=Stealth] (5.75,3.75) -- (5.75,15.5);
\draw [line width=0.8pt, ->, >=Stealth] (5.75,3.75) -- (16.25,3.75);
\draw [line width=0.8pt, short] (8.5,8.5) -- (15.75,14.6);
\draw [line width=0.8pt, short] (8.5,8.5) -- (4.75,11.9); 
\draw [line width=0.8pt, dashed] (5.75,5) -- (16.25,5);
\draw [line width=0.8pt, dashed] (5.75,6.75) -- (16.25,6.75);
\draw [line width=0.8pt, dashed] (5.75,8.5) -- (16.25,8.5);
\draw [line width=0.8pt, dashed] (5.75,10.25) -- (16.25,10.25);
\draw [line width=0.8pt, dashed] (5.75,12) -- (16.25,12);
\draw [line width=0.8pt, dashed] (5.75,13.75) -- (16.25,13.75);
\node [] at (15.5,3.2) {space};
\node [] at (5,14.5) {time};
\node [] at (15,9.15) {$t = t_0$};
\node [] at (15,10.9) {$t = t_1$};
\node [] at (15,12.65) {$t = t_2$};
\draw [ color={rgb,255:red,0; green,66; blue,170}, line width=1.4pt, short] (8.75,3.75) .. controls (8,11.25) and (9.25,11) .. (8.5,15.25);
\draw [ color={rgb,255:red,0; green,66; blue,170}, line width=1.4pt, short] (13,3.75) .. controls (14.5,9.75) and (11.5,10.5) .. (13,15.25);
\draw [ color={orange}, line width=5pt, short] (13.3,8.5) .. controls (13.2,9.3) .. (12.9,10.25);
\fill (8.6,6.75) circle (4pt); \node[below, yshift=-6mm] {};
\fill (13.4,6.75) circle (4pt); \node[below, yshift=-6mm] {};
\draw[thick, decorate, decoration={snake, amplitude=3mm, segment length=10mm, post length=0.1mm}] (8.6,6.75) -- (13.4,6.75);
\node [] at (11,6)  {$\ket{\psi}$};
\node [] at (8.6,3.2) {Alice};
\node [] at (13,3.2) {Bob};
\fill (11,4.6) circle (4pt); \node[below, yshift=-6mm] {};
\node [] at (11,4.15) {source};
\draw [line width=0.8pt, short] (11,4.6) -- (15.75,8.8);
\draw [line width=0.8pt, short] (11,4.6) -- (4.75,10);
\fill (8.5,8.5) circle (4pt); 
\fill (12.9,10.25) circle (4pt); 
\fill (13.3,8.5) circle (4pt); 
\end{circuitikz}
}
\resizebox{0.54\textwidth}{!}{
\begin{circuitikz}
\node [] at (4,3) {\Huge{\textbf{b)}}};
\tikzstyle{every node}=[font=\LARGE]
  \draw[pattern={horizontal lines},pattern color=orange,draw=none]
    (5,6.75) rectangle +(11.5,3.5);
\draw [line width=0.8pt, ->, >=Stealth] (5.75,3.75) -- (5.75,15.5);
\draw [line width=0.8pt, ->, >=Stealth] (5.75,3.75) -- (16.25,3.75);
\fill (8.5,8.5) circle (4pt); 
\fill (12.9,10.25) circle (4pt); 
\draw [line width=0.8pt, short] (8.5,8.5) -- (15.75,14.6);
\draw [line width=0.8pt, short] (8.5,8.5) -- (4.75,11.9);
\draw [line width=0.8pt, dashed] (5.75,5) -- (16.25,5);
\draw [line width=0.8pt, dashed] (5.75,6.75) -- (16.25,6.75);
\draw [line width=0.8pt, dashed] (5.75,8.5) -- (16.25,8.5);
\draw [line width=0.8pt, dashed] (5.75,10.25) -- (16.25,10.25);
\draw [line width=0.8pt, dashed] (5.75,12) -- (16.25,12);
\draw [line width=0.8pt, dashed] (5.75,13.75) -- (16.25,13.75);
\node [] at (15.5,3.2) {space};
\node [] at (5,14.5) {time};
\node [] at (15,9.15) {$t = t_0$};
\node [] at (15,10.9) {$t = t_1$};
\node [] at (15,12.65) {$t = t_2$};
\draw [ color={rgb,255:red,0; green,66; blue,170}, line width=1.4pt, short] (8.75,3.75) .. controls (8,11.25) and (9.25,11) .. (8.5,15.25);
\draw [ color={rgb,255:red,0; green,66; blue,170}, line width=1.4pt, short] (13,3.75) .. controls (14.5,9.75) and (11.5,10.5) .. (13,15.25);
\fill (8.6,6.75) circle (4pt); \node[below, yshift=-6mm] {};
\fill (13.4,6.75) circle (4pt); \node[below, yshift=-6mm] {};
\draw[thick, decorate, decoration={snake, amplitude=3mm, segment length=10mm, post length=0.1mm}] (8.6,6.75) -- (13.4,6.75); 
\node [] at (11,6)  {$\ket{\psi}$};
\node [] at (8.6,3.2) {Alice};
\node [] at (13,3.2) {Bob};
\fill (11,4.6) circle (4pt); \node[below, yshift=-6mm] {};
\node [] at (11,4.15) {source};
\draw [line width=0.8pt, short] (11,4.6) -- (15.75,8.8);
\draw [line width=0.8pt, short] (11,4.6) -- (4.75,10);
\end{circuitikz}
}
\end{center}
\vspace*{-1.5cm}
\caption{
\textbf{Spacetime diagrams illustrating the protocol for superluminal signalling}. The dashed lines show a preferred time-slicing. Alice and Bob travel along future-directed time-like curves (thin blue lines). Initially, they share a state $\ket{\psi}$ generated at a source in their common past. \textbf{a)} In the scenario considered in \cite{Gisin2001}, Alice makes a local quantum measurement at $t = t_0$ thus preparing some local state for Bob. Bob puts the later in his device, which operates until $t=t_1$ (thick orange line).  
\textbf{b)} In our scenario the evolution of the state $\ket{\psi}$ is global, which is depicted by thin orange lines parallel to the time slices --- see Section \ref{sec:bi} for the description. The events associated with Alice's and Bob's local operations are marked with black dots.}
\label{fig:Gisin}
\end{figure}
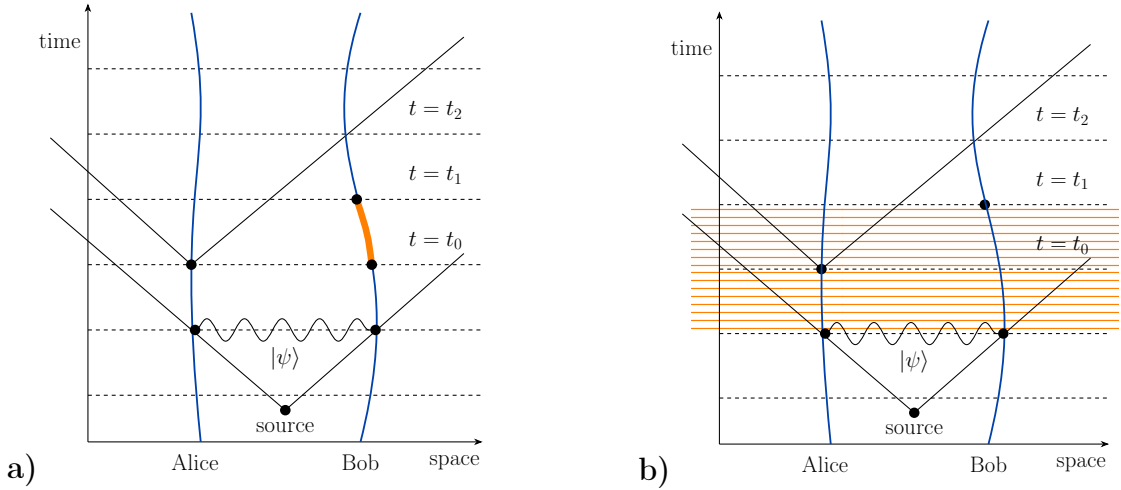

This brings us to a major issue, which lies beyond the scope of Gisin's theorem: It is not specified how Bob's device would behave if Alice did \emph{not} make a measurement.  In order to address this question one would need an extension of the map $\E: \h_B \to S(\h_B)$ to a map $\widetilde{\E}: \h_A \otimes \h_B \to \h_A \otimes \h_B$ or $\widetilde{\E}': \h_A \otimes \h_B \to S(\h_A \otimes \h_B)$. Alternatively, though not quite equivalently, one would need to specify an extension of $\E$ to mixed states, $\overline{\E}: S(\h_B) \to S(\h_B)$. 
It is unclear how to define such extensions consistently, unless $\E$ is a linear map.
Consequently, the use of Gisin's theorem is limited to \emph{local} dynamical maps acting on \emph{pure} states only. 
It  does not apply neither to dynamics $\overline{\E}$ defined directly on mixed states $S(\h_A)$, nor to nonlocal dynamics $\widetilde{\E}$, $\widetilde{\E}'$ defined globally on $\h_A \otimes \h_B$.
 In \cite{Caban20,Caban21} the authors provided explicit examples of \emph{nonlinear} maps $\overline{\E}: S(\h_B) \to S(\h_B)$, which bypass the conclusions of Gisin's theorem. Another example --- an extension of the Schr\"odinger--Newton equation --- was recently considered in \cite{Paterek24}.

Finally, the theorem presented in \cite{Gisin2001} makes some assumptions concerning the agents' freedom to act. Concretely:

\begin{itemize}
    \item[(4)] Alice is free to choose \emph{any} local POVM, so that she can prepare \emph{any} ensemble of local states for Bob.
    \item[(5)] Bob can reconstruct his local density operator to an arbitrary precision.
\end{itemize}
Assumption (4) means that not only Alice has to be able to choose any local observable on $\h_A$, but also that she is able to couple her system to an ancilla $A'$ and perform a projective measurement on $\h_A \otimes \h_{A'}$.

Assumption (5) is milder --- it requires Bob to be able to measure any projector from a tomographically complete set. In the context of General Probabilistic Theories this assumption is known under the name of the ``local tomography axiom'' \cite{Chiribella10}.
Clearly, if we assume that Alice and Bob have access to the same resources, then we have to adopt the stronger condition (4) for both Alice and Bob. 

One can also read these two assumptions as consequences of the ``no-restriction hypothesis'' put forward within General Probabilistic Theories \cite{Chiribella10,No_restrictions}. This hypothesis states, roughly, that \cite{No_restrictions}: ``any mathematically well-defined measurement in the theory should be physically allowed.''  It is, in fact, a rather strong assumption, which excludes large classes of beyond-quantum theories \cite{No_restrictions}. The no-restriction hypothesis is also violated in the scheme of \cite{Kaplan2022}, where the authors assume that the position basis is distinguished in Nature. 

In this context, let us emphasise that the `restriction' refers to the constraints imposed on the set of observables, which are determined by Nature. These constraints can, in principle, be different for Alice and Bob, for instance, because they operate in different regions of spacetime. Notwithstanding, we assume that any agent is free to chose any observable from the set determined by Nature.

This concludes our analysis of Gisin's argument against \emph{local} nonlinear quantum dynamics. We shall now pass on to a detailed analysis of a class of \emph{global} nonlinear quantum dynamics.

\section{General considerations}
\label{sec:3}
\subsection{A class of nonlinear quantum dynamics}
Let $\h$ be a Hilbert space of finite dimension $N+1$ and let $\B(\h)$ denote the space of all hermitian operators on $\h$. Suppose now that the dynamics on $\h$ is governed by a 
nonlinear Schr\"odinger equation:
\begin{equation}\label{eq-gen}
    i\frac{d}{dt} \ket{\psi(t)} = H \ket{\psi(t)} + K\big( \ket{\psi(t)} \big),
\end{equation}
where $H \in \s(\h)$ is a Hamiltonian of the system and $K$ is a `self-potential', that is some nonlinear map $K: \h \to \h$.

We shall consider a specific class of self-potentials, which \emph{distinguishes a chosen basis} of $\h$, $\big\{\ket{j}\big\}_{j=0}^{N}$. Concretely,
\begin{equation}
 K \big( \ket{\psi} \big) =  \sum_{j=0} ^{N} f_j \big( \left| \bra{\psi} A_j \ket{j} \right| \big) \braket{j}{\psi} \ket{j},
 \label{general-nonlinearity}
\end{equation}
with some real functions $f_j$  and operators $A_j \in \s(\h)$. A self-potential of the form \eqref{general-nonlinearity} can be seen as a discrete version of a continuous-variable nonlinear Schr\"odinger equation \cite{MeyerWong13,MeyerWong14,Childs}. For instance, \eqref{general-nonlinearity} with $A_j=\bone$ and $f_j(x)= g x^2$, for some parameter $g\in\R$, yields a nonlinear equation 
\begin{equation}
    i\frac{d}{dt} \ket{\psi(t)} = H \ket{\psi(t)} + g \sum^{N}_{j=0}\left|\braket{j}{\psi(t)}\right|^2 \braket{j}{\psi(t)} \ket{j}.
    \label{square}
\end{equation}
This is a discrete version of the Gross--Pitaevskii equation
\begin{equation}
    i\frac{d}{dt} \psi(x,t) = H \psi(t,x) + g \left|{\psi}(t,x)\right|^2 \psi(t,x),
    \label{squareX}
\end{equation}
which explicitly distinguishes the position (improper) basis in the infinite-dimensional Hilbert space $\h = L^2(\R^3)$. Along the same lines, one can consider e.g. a discrete analogue of the Bia{\l}ynicki-Birula--Mycielski equation \cite{IBB_nonlinear},
\begin{equation}
    i\frac{d}{dt} \ket{\psi(t)} = H \ket{\psi(t)} + g \sum^{N}_{j=0} \big( \log \left|\braket{j}{\psi(t)}\right| \big) \braket{j}{\psi(t)} \ket{j}.
    \label{log}
\end{equation}

Nonlinear equations with a self-potential of the form \eqref{general-nonlinearity} do not preserve the scalar products, but they do preserve the norm of state vectors  (cf. \cite{IBB_nonlinear}). Indeed, 
\begin{equation}
    \begin{split}
         \frac{d}{dt} \braket{\psi}{\varphi} & =
         \bra{\psi} \left(\frac{d}{dt}\ket{\varphi}\right) + \left( \frac{d}{dt} \bra{\psi} \right)\ket{\varphi}\\
         &= i \sum^{N}_{j=0} \Big[f_j \big( \left| \bra{\psi} A_j \ket{j} \right| \big) - f_j \big( \left| \bra{\varphi} A_j \ket{j} \right| \big) \Big] \braket{\psi}{j}\braket{j}{\varphi},
    \end{split}
    \label{dadt}
\end{equation}
which equals zero for $\ket{\varphi} = \ket{\psi}$. 

As in the standard Schr\"odinger equation, the Noether symmetry associated with the probability conservation is the $U(1)$ symmetry, $\ket{\psi(t)} \rightsquigarrow  e^{i \lambda \ket{\psi(t)}}$, for $\lambda \in \mathbb{R}$. One can make the global phase time-dependent by a suitable compensation in the Hamiltonian part:
\begin{align}\label{gauge}
 \ket{\psi(t)} \rightsquigarrow  e^{i \lambda(t)}\ket{\psi(t)}, \qquad H \rightsquigarrow  H - \lambda'(t) \bone.
\end{align}

\subsection{Nonlinear dynamics in bipartite systems\label{sec:bi}}

The central goal of this paper is to study the dynamics of quantum information in spacelike-separated bipartite systems under a nonlinear evolution of type \eqref{eq-gen}. To this end we consider a tensor product Hilbert space, $\h = \h_A \otimes \h_B$, with the basis $\{ \ket{jk} \}$, for $j \leq \dim \h_A$, $k \leq \dim \h_B$, distinguished by the nonlinear term \eqref{general-nonlinearity}. In general, the local Hilbert spaces need not have the same dimension. Eq.~\eqref{eq-gen} thus becomes
\begin{equation}\label{eq-gen2}
    i\frac{d}{dt} \ket{\psi(t)} = H \ket{\psi(t)} + \sum_{j,k} f_{jk} \big( \left| \bra{\psi} A_{jk} \ket{jk} \right| \big) \braket{jk}{\psi} \ket{jk},
\end{equation}
where $A_{ij}$ is an operator on a tensor product on a Hilbert space. An elementary assumption we need to make in order to avoid superluminal signalling (see \cite{LocalH1,LocalH2}) is the locality of the Hamiltonian part, $H = H_A\otimes \mathbb{1} +  \mathbb{1} \otimes H_B$. On the other hand, locality \emph{cannot} be imposed on the self-potential, because the nonlinear term in \eqref{eq-gen2} acts on global states. In other words, the nonlinearity of type \eqref{general-nonlinearity} is always nonlocal.

Equation \eqref{eq-gen2} defines a dynamical map $\F: \R^+ \times \h_A \otimes \h_B \to \h_A \otimes \h_B$, which describes the time-evolution in the space of pure states, $\ket{\Psi(t)} = \ket{ \F \big(t,\Psi(0)\big)}$. Such a map determines uniquely the  dynamics in the spaces of local states, $S(\h_A)$ and $S(\h_B)$, through the partial trace rule:
\begin{align}
\rho_{A/B}(t) \vc \Tr_{B/A} \ket{\Psi(t)} \bra{\Psi(t)}.
\label{eq:partial-trace}
\end{align}
Note, however, that the function $t \mapsto \rho_{A/B}(t)$ \textit{does not}, in general, define a dynamical map $\overline{\E}: \R^+ \times S(\h_B) \to S(\h_B)$, because $\rho_{A/B}(t)$ depends on the global state $\ket{\Psi(t)}$ (cf. \cite{Alicki95}).

With the vector state decomposition,
\begin{equation}
    \ket{\psi}= \sum_{j,k} \alpha_{jk} \ket{jk}, \qquad \sum_{j,k} |\alpha_{jk}|^2 =1,
    \label{two-qubit-state}
\end{equation}
the evolution is governed by a set of coupled nonlinear equations on $N+1$ complex functions $\alpha_{jk}$,
\begin{equation}
i\partial_t \alpha_{jk}(t) = \sum_{\ell,m} \alpha_{\ell m}(t) \bra{jk} H_A\otimes \mathbb{1} +  \mathbb{1} \otimes H_B \ket{\ell m} + f_{jk} \left( \Big\vert \sum_{n,o} \alpha_{no}(t) \bra{jk} A_{jk} \ket{no}  \Big\vert \right) \alpha_{jk}(t). 
\label{G-P-system}
\end{equation}

For the sake of simplicity, we restrict ourselves to the case $A_{jk} = \bone$, so that the coupling between the equations for different $jk$'s is effectuated only by the local Hamiltonian term:
\begin{equation}
i\partial_t \alpha_{jk}(t) = \sum_{\ell} \alpha_{\ell k}(t) \bra{j} H_A \ket{\ell} + \sum_{m} \alpha_{jm}(t) \bra{k} H_B \ket{m} + f_{jk} \left( \left\vert \alpha_{jk}(t)  \right\vert \right) \alpha_{jk}(t). 
\label{G-P-system2}
\end{equation}
Suppose now that two spacelike separated parties, Alice and Bob, share a --- possibly entangled --- quantum state $\ket{\psi}$. The state undergoes a global time-evolution governed by equation \eqref{eq-gen}. The time-parameter $t$ appearing explicitly in Eq. \eqref{eq-gen} is a `laboratory time' corresponding to a global slicing of the Minkowski spacetime --- see Fig. \ref{fig:Gisin} \textbf{b)}. Let us emphasise that the preferred time-slicing of spacetime is determined explicitly by dynamical equation \eqref{eq-gen} and not implicitly by the measurement (recall assumption (3b)), as it was the case in Gisin's scheme. 

This should be contrasted with the setting of Gisin's scenario, in which Alice prepares at a distance an ensemble of \emph{local} pure quantum states for Bob, who puts them into a \emph{local} device and then measures the outgoing state. In such a scenario, we ignore the Alice's part of the state and actually we have to do so, because there is no natural way of extending the action of Bob's local box on a global state.

Hence, the advantage of working with a global nonlinearity of the form \eqref{eq-gen} is that we avoid the notorious problem of extending the dynamics from pure states to statistical mixtures. The price to pay is that we are allowed to consider global pure states only, also after a local measurement. Consequently, in contrast to \cite{Gisin2001}, we do need to assume some state-update rule when implementing the local measurement. On the other hand, we will also consider a communication protocol, which does not require any measurement on Alice's side. Another advantage of the scenario considered here is that, while we do need to make assumption (3b), Bob's time of measurement need not be synchronised with the time of Alice's intervention. In particular, in Sec. \ref{sec:meas} we consider a communication protocol, in which Alice makes or does not make a measurement.

In summary, we adopt assumptions (1a--1d) and (2) from Section \ref{sec:2}, and we wish to check whether it is possible to satisfy (3a) --- the no-signalling principle --- within the class of global nonlinear dynamics \eqref{eq-gen}. To that end, we shall consider various relaxations of the no-restriction hypothesis and explore the possibility of signalling through local unitaries rather than projective measurements.

\subsection{Analytic solutions}

Eqs. \eqref{G-P-system2} form a complicated system of coupled complex nonlinear ODEs with parameters. In general, one has to resort to numerical methods to study their behaviour. There exists, however, a special class of dynamics \eqref{G-P-system2}, for which an explicit analytic solution can be written.

\begin{prop}
   If both local Hamiltonians $H_A$, $H_B$ are diagonal in the basis singled out by the nonlinearity, then the system \eqref{G-P-system2} has a unique analytic solution. 
   \label{prop:sol}
\end{prop}
\begin{proof}
Let us denote 
\begin{align}\label{X}
    a_j \vc \bra{j} H_A \ket{j}, && b_k \vc \bra{k} H_B \ket{k}, && \text{ and } && X_{jk} \vc \bra{jk} H_A \otimes \bone + \bone \otimes H_B \ket{jk} = a_j + b_k.
\end{align}
Because the local Hamiltonians are diagonal in the preferred basis, so is the global Hamiltonian. The system \eqref{G-P-system2} decouples and we are left with $(\dim \h_A) \cdot (\dim \h_B)$ independent ODEs:
\begin{equation}
    i \partial_t \alpha_{jk}(t) = \alpha_{jk}(t)X_{jk} + f_{jk}\left(\vert \alpha_{jk}(t) \vert \right) \alpha_{jk}(t), 
    \label{system-of-solutions-diagonal}
\end{equation}
Each of these equations can be solved analytically. Let us write 
\begin{equation}\label{ansatz}
    \alpha_{jk}(t) = \aleph_{jk}(t) e^{i \varphi_{jk}(t)},
\end{equation}
where $\aleph_{jk}$ and $\varphi_{jk}$ are real functions. Then, \eqref{system-of-solutions-diagonal} yields
\begin{equation}\label{eq_diag}
   i \left(\aleph'_{jk}(t) e^{i \varphi_{jk}(t)} + i \varphi'_{jk}(t) \aleph_{jk}(t) e^{i \varphi_{jk}(t)} \right) = X_{jk} \aleph_{jk}(t) e^{i \varphi_{jk}(t)} + f_{jk}\left(\vert\aleph_{jk}(t) \vert\right) \aleph_{jk}(t) e^{i \varphi_{jk}(t)}.
\end{equation}
We thus get: 
\begin{equation}
    \begin{cases} 
    \aleph'_{jk}(t) = 0, \\
    \varphi'_{jk}(t) = - \left(X_{jk} + f_{jk}\left(\vert\aleph_{jk}(t) \vert\right)\right). 
 \end{cases}
\end{equation}
Hence, the solutions are 
\begin{equation}\label{sol}
    \begin{cases}
        \aleph_{jk}(t) = \aleph_{jk}^0=\text{const.}, \\
        \varphi_{jk}(t) = - \left(X_{jk} + f_{jk}(\aleph_{jk}^0)\right)t+\varphi_{jk}^0, 
    \end{cases}
\end{equation}
where $\alpha_{jk}(0)=\aleph_{jk}^0 e^{i \varphi^0_{jk}}$ are the initial conditions for \eqref{G-P-system2}.
\end{proof}

It turns out that if one of the local Hamiltonians is diagonal in the preferred basis, then the system \eqref{G-P-system2} has an additional symmetry.

\begin{prop}\label{prop:diag}
    If $H_B$ is diagonal in the basis singled out by the nonlinearity, then 
    \begin{align}
        \bra{k} \rho_B(t) \ket{k} = \bra{k} \rho_B(0) \ket{k}, \quad \text{ for all } \quad t \geq 0 \; \text{ and } \;  k \in \{0,\ldots,\dim \h_B\}.
    \end{align}
    Symmetrically, if $H_A$ is diagonal, then the diagonal of $\rho_A$ is constant in time.
\end{prop}
\begin{proof}
    With the notation \eqref{two-qubit-state} we have $\bra{k} \rho_B \ket{k'} = \sum_{j} \alpha_{jk} \overline{\alpha_{jk'}}$. Using equations \eqref{G-P-system2} we compute
    \begin{equation}
    \begin{aligned}
        \dt_t \bra{k} \rho_B(t) \ket{k'} & = \sum_j \big[ ( \dt_t \alpha_{jk}(t) ) \overline{\alpha_{jk'}(t)} + \alpha_{jk}(t)  \dt_t \overline{\alpha_{jk'}(t)}\, \big] \\
        & = i \sum_j \bigg[ \sum_{\ell} \Big( \overline{\alpha_{\ell k}(t)} \overline{\bra{j} H_A \ket{\ell}} \alpha_{jk'}(t) - \alpha_{\ell k}(t) \bra{j} H_A \ket{\ell} \overline{\alpha_{jk'}(t)} \Big) \\
        & \quad + \sum_{m} \Big( \overline{\alpha_{jm}(t)} \overline{\bra{k} H_B \ket{m}} \alpha_{jk'}(t) - \alpha_{jm}(t) \bra{k} H_B \ket{m} \overline{\alpha_{jk'}(t)} \Big) \\
        & \quad + f_{jk} \left( \left\vert \alpha_{jk}(t)  \right\vert \right) \overline{\alpha_{jk}(t)} \alpha_{jk'}(t) - f_{jk} \left( \left\vert \alpha_{jk}(t)  \right\vert \right) \alpha_{jk}(t) \overline{\alpha_{jk'}(t)}  \, \bigg].
        \end{aligned}
    \end{equation}
    If $\bra{k} H_B \ket{m} = b_k \delta_{km}$, with some $b_k \in \R$, then we get
    \begin{align}
        \dt_t \bra{k} \rho_B(t) \ket{k} & = i \sum_{j,\ell} \Big( \overline{\alpha_{\ell k}(t)} \bra{\ell} H_A \ket{j} \alpha_{jk}(t) - \alpha_{\ell k}(t) \bra{j} H_A \ket{\ell} \overline{\alpha_{jk}(t)} \Big) = 0,
    \end{align}
    where we have used the fact that one can exchange the indices $j$ and $\ell$, which both range from 0 to $\dim \h_A - 1$.
    
    Analogously, one shows that the diagonal of $\rho_A$ is constant when $\h_B$ is diagonal. 
\end{proof}

\section{Nonlinear dynamics of two qubits\label{sec:4}} 

\subsection{Global dynamics}

For our numerical studies, we shall consider the simplest non-trivial case of the system \eqref{G-P-system2}, which is that of two qubits, $\dim(H_A)=\dim(H_B)=2$. We shall also focus on the Gross--Pitaevskii type nonlinearity \eqref{square}. We parametrise the local Hamiltonian as
\begin{align}\label{local-hamiltonian}
   H_A=\begin{pmatrix}
a_1 & c \\
\bar{c} & a_2
\end{pmatrix}, 
&&
    H_B=\begin{pmatrix}
b_1 & d \\
\bar{d} & b_2
\end{pmatrix},
\end{align}
with free parameters $a_1, a_2, b_1, b_2 \in \R$ and $c, d \in \mathbb{C}$.  The system \eqref{G-P-system2} then reads
\begin{equation}
i \partial_t 
\begin{pmatrix}
 \alpha_{00}(t) \\
\alpha_{01}(t) \\
 \alpha_{10}(t) \\
\alpha_{11}(t)
\end{pmatrix} = 
 \begin{pmatrix}
a_1+b_1 & d & c & 0 \\
\bar{d} & a_1+b_2 & 0 & c \\
\bar{c} & 0 & a_2+b_1 & d \\
0 & \bar{c} & \bar{d} & a_2+b_2 
\end{pmatrix} 
\begin{pmatrix}
 \alpha_{00}(t) \\
\alpha_{01}(t) \\
 \alpha_{10}(t) \\
\alpha_{11}(t)
\end{pmatrix}
+ g
\begin{pmatrix}
\vert \alpha_{00}(t) \vert^2 \alpha_{00}(t) \\
\vert \alpha_{01}(t)  \vert^2  \alpha_{01}(t) \\
\vert \alpha_{10}(t)  \vert^2  \alpha_{10}(t) \\
\vert \alpha_{11}(t)  \vert^2  \alpha_{11}(t)
\end{pmatrix}.
\label{GP_qubits}
\end{equation}

Suppose that Alice and Bob share an entangled state of two qubits $\ket{\psi}$. The state $\rho_B$ effectively available to Bob is determined by the standard trace rule (\ref{eq:partial-trace}).

In the preferred basis $\{\ket{jk}\}$ it has the following matrix presentation: 
\begin{equation}
    \rho_B=  \begin{pmatrix}
        \abs{\alpha_{00}}^2 + \abs{\alpha_{10}}^2 & \alpha_{00}\overline{\alpha_{01}} + \alpha_{10}\overline{\alpha_{11}} \\ 
        \alpha_{01}\overline{\alpha_{00}} + \alpha_{11}\overline{\alpha_{10}} & \abs{\alpha_{01}}^2 + \abs{\alpha_{11}}^2
    \end{pmatrix}.
\end{equation}
On the other hand, any two-dimensional density operator $\rho$ can be expanded using Pauli matrices $\vec{\sigma}=(\sigma_x,\sigma_y,\sigma_z)$: 
\begin{equation}
    \rho=\frac{1}{2} (\mathbb{1}+\vec{n}\cdot \vec{\sigma}) = 
    \frac{1}{2} \begin{pmatrix}
        1 + n_z & n_x - i n_y \\ 
        n_x + i n_y & 1-n_z
    \end{pmatrix},
    \label{rho}
\end{equation}
where $\vec{n}=(n_x,n_y,n_z)$ is a Bloch vector within the Bloch ball. Using formula \eqref{rho} we get 
\begin{align}
    \begin{cases}
        n_x&= 2 \Re \left(\alpha_{00}\overline{\alpha_{01}} + \alpha_{10}\overline{\alpha_{11}} \right), \\
        n_y&= - 2 \Im \left(\alpha_{00}\overline{\alpha_{01}} + \alpha_{10}\overline{\alpha_{11}} \right), \\
        n_z&= 2 \left(\vert \alpha_{00} \vert^2 + \vert \alpha_{10} \vert^2\right) - 1. 
    \end{cases}
    \label{Bloch-vector-components}
\end{align}
Furthermore, we have
\begin{align}\label{conc}
\Tr \rho^2_B = \tfrac{1}{2}\big( 1 + \Vert \vec{n} \Vert^2 \big), \qquad \text{and} \qquad \C(\ket{\psi}) = \sqrt{2\big(1-\Tr \rho^2_B \big)} = \sqrt{1- \Vert \vec{n} \Vert^2},
\end{align}
where $\C$ is the \emph{concurrence} -- a measure of entanglement  between the two qubits. The state $\ket{\psi}$ is separable if and only if $\rho_B$ is pure, that is $\Vert \vec{n} \Vert = 1$ and it is maximally entangled iff $\rho_B = \tfrac{1}{2} \bone$.

\subsection{\label{sec:global} Distinguishability of global states}

Let us start with exploring some general characteristics of the two-qubit nonlinear dynamics. Eqs. \eqref{GP_qubits} form a complicated system of 4 coupled complex ODEs with 8 free real parameters: 7 coefficients of the local Hamiltonian (one can be gauged away \eqref{gauge}) and the nonlinear coupling $g$. Hence, a complete analysis of the full parameter space is not possible. Instead, we shall present, in this and the forthcoming section, the highlights from our numerical studies with some fixed parameters. 

A detailed study of nonlinear dynamics \eqref{eq-gen} of a single qubit was carried out in \cite{Childs}. The authors showed that such a system exhibits exponential sensitivity to the initial conditions. More concretely, two initial states with a small overlap, $|\braket{\psi(0)}{\phi(0)}| = 1- \varepsilon$, become distinguishable, $|\braket{\psi(t)}{\phi(t)}| \approx 0$, for times $t = O(g^{-1} \log \varepsilon^{-1})$. This was shown to hold for a general class of nonlinear functions including, in particular, the Gross-Pitaevskii \eqref{square} and logarithmic nonlinearities \eqref{log}. One can expect the system \eqref{GP_qubits} to exhibit similar behaviour. 

Let us start with a simple example of `purely nonlinear' dynamics with $H_A = H_B = 0$. In such a case, an analytic solution to Eqs. \eqref{GP_qubits} is given by formulae \eqref{sol} with $X_{jk}=0$.
Let us consider a class of initial states:
\begin{align}\label{psi0}
    \ket{\Psi_{x}} = \tfrac{1}{\sqrt{2(1+x^2)}} \big( (1+x)\ket{00} + (1-x) \ket{11} \big), \quad \text{ for } x \in [0,1].
\end{align}
For $x=0$, $ \ket{\Psi_{x}}$ is a Bell state, while for $x=1$ it becomes a separable state. More generally, we have $\C(\ket{\Psi_{x}})=(1-x^2)/(1+x^2)$. 

Let us consider two initial states, $\ket{\psi(0)}$ and $\ket{\phi(0)}$, from the class \eqref{psi0} such that $\ket{\psi(0)}$ has a fixed value of concurrence $\C$ and $\ket{\phi(0)}$ is chosen so that $\vert \braket{\psi(0)}{\phi(0)} \vert = 1- \varepsilon$, for some small $\varepsilon>0$.

Figure \ref{fig:overlap1} illustrates the dynamics of the overlap between these two states, $d(t) = |\braket{\psi(t)}{\phi(t)}|$, for a fixed value of concurrence $\C$. It shows that, with $H_A = H_B = 0$, the distinguishability is best --- both in terms of maximal distinguishability and in terms of time to achieve it --- if the initial state is maximally entangled. If one of the initial states is separable, then the distinguishability is constant. 

\begin{figure}[h]
\includegraphics[scale=1]{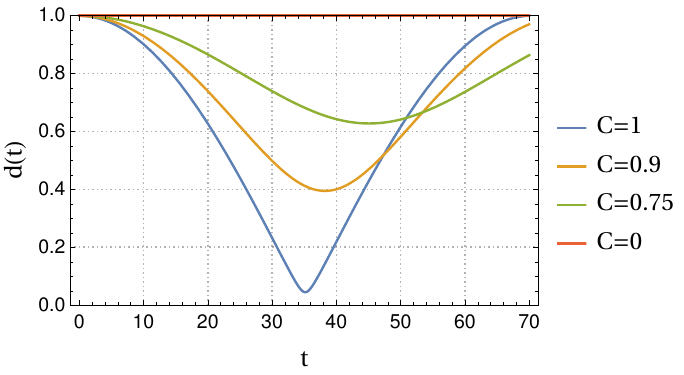}
\caption{\label{fig:overlap1}\textbf{The time-evolution of the overlap between two global states} initially separated by $\varepsilon = 0.001$ for different values of the concurrence $\C$ of the initial state $\ket{\psi(0)}$. The parameters of the dynamics \eqref{GP_qubits} are $H_A = H_B = 0$ and $g=1$.}
\end{figure}

The distinguishability of global states can be significantly enhanced by tuning the parameters of the local dynamics, as illustrated in Figure \ref{fig:overlap2}. The plot on the logarithmic scale, Fig. \ref{fig:overlap2} \textbf{b)}, suggests that the overlap between the states can decay exponentially fast and that the dynamics of the overlap is, in general, chaotic. This is coherent with the results of \cite{Childs} for a single qubit nonlinear dynamics.

\begin{figure}[h]
\textbf{a)} \hspace*{-0.3cm}
\resizebox{0.47\textwidth}{!}{\includegraphics[scale=1]{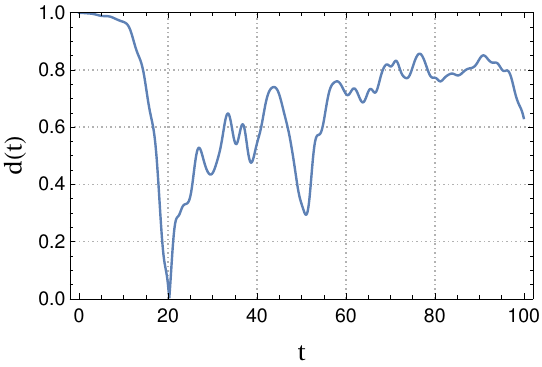}}
\qquad\textbf{b)}\hspace*{-0.5cm}
\resizebox{0.46\textwidth}{!}{\includegraphics[scale=1]{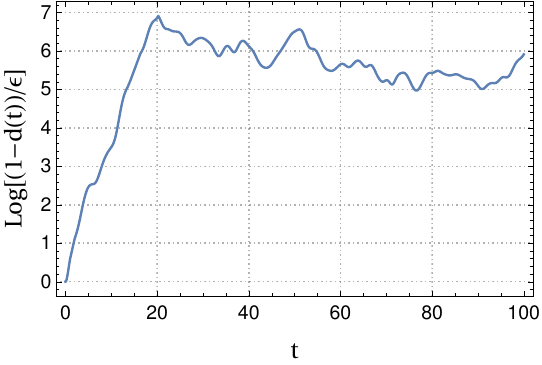}}
\caption{\label{fig:overlap2} \textbf{The time-evolution of the overlap between two states} under the dynamics \eqref{GP_qubits} with $a_1 = 0.1$, $c=0.413$, $d=0.108$, $g=1$ and $a_2=b_1=b_2=0$. The initial state $\ket{\psi(0)} = (\ket{00}+\ket{11})/\sqrt{2}$ is maximally entangled state and $\varepsilon = 0.001$. The plots are shown on \textbf{a)} the linear scale, and \textbf{b)} the logarithmic scale.}
\end{figure}

\subsection{Dynamics of entanglement}

In general, the entanglement in the system described by Eqs. \eqref{GP_qubits} is not constant. In fact, the concurrence \eqref{conc} may fluctuate in time, even if there is no local dynamics, i.e. $H_A= H_B = 0$. Hence, not surprisingly, the nonlinear quantum dynamics of the form \eqref{G-P-system2} goes beyond the LOCC paradigm \cite{Horodecki}, within which the entanglement can never increase in time.

In Fig. \ref{fig:concurrence} we present exemplary dynamics of concurrence under the dynamics \eqref{GP_qubits}. First, we take two initial entangled states, $\ket{\psi(0)}$ and $\ket{\phi(0)}$, from the class \eqref{psi0} such that $\ket{\psi(0)}$ is maximally entangled and $\C(\ket{\phi(0)}) = 1 - \varepsilon$, for a small $\varepsilon >0$. Fig. \ref{fig:concurrence} \textbf{a)} shows that the evolution of concurrence of these two states quickly bifurcates.

Secondly, we take two initial separable states,
\begin{align}\label{conc_sep}
    \ket{\psi(0)} = \ket{00} \qquad \text{ and } \qquad \ket{\phi(0)} = \tfrac{1}{\sqrt{1 + 2(\varepsilon-1)\varepsilon}} \ket{0} \big((1-\varepsilon) \ket{0} + \varepsilon \ket{1})\big).
\end{align}
Fig. \ref{fig:concurrence} \textbf{b)} shows that entanglement can arise under the dynamics \eqref{GP_qubits} even if one starts from a separable initial state. The dynamics of concurrence is again sensitive to the initial state.

\begin{figure}[h]
    \centering
    \textbf{a)} \resizebox{0.44\textwidth}{!}{\includegraphics[scale=1]{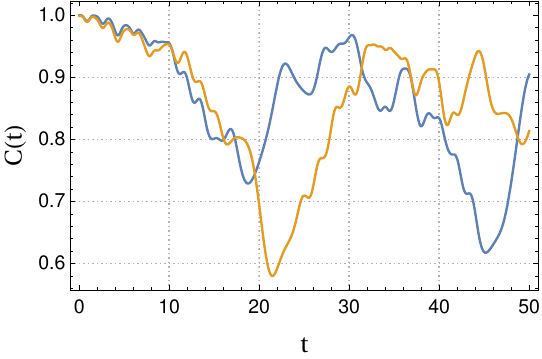}} \qquad
    \textbf{b)} \resizebox{0.44\textwidth}{!}{\includegraphics[scale=1]{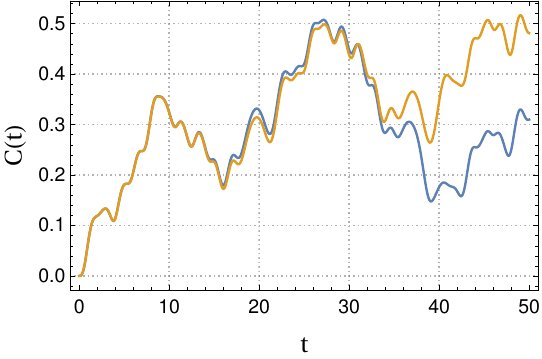}}
    \caption{\textbf{The time-evolution of the concurrence} for system \eqref{GP_qubits} with parameters $a_1 = 0.2$, $c=0.4$, $d=0.5$, $g=1$ and $a_2=b_1=b_2=0$. \textbf{a)} The initial states are the Bell state (blue curve) and a state \eqref{psi0} with concurrence $1-\varepsilon$, for $\varepsilon = 0.001$ (orange curve). \textbf{b)} The case of separable initial states \eqref{conc_sep} with $\varepsilon = 0.001$.}
    \label{fig:concurrence}
\end{figure}

\subsection{\label{sec:global}Sensitivity of local dynamics on the global state}

In the previous subsections we have studied the dynamics of global quantities impelled by the equation \eqref{GP_qubits}. We have witnessed the sensitivity of the overlap and concurrence on the choice of the initial state, at least for some values of dynamics' parameters. We shall now provide further evidence for chaos in the system  \eqref{GP_qubits}. 

Let us consider the dynamics of Bob's local effective state $\rho_B(t)$, as defined by Eq.\eqref{eq:partial-trace}. Using the Fano parametrisation \eqref{rho}, we can depict the time-evolution of $\rho_B$ as a trajectory on a Bloch ball. Figure \ref{fig:Bloch0} shows an example of such a trajectory for dynamics \eqref{GP_qubits} with parameters $a_1 = c = 0.1$, $d=0.3$, $g=2$ and $a_2=b_1=b_2=0$, for the initial state  $\ket{\psi(0)} = (\ket{00}+\ket{11})/\sqrt{2}$ and the time range $t \in [0,100]$.

\begin{figure}[h]
    \textbf{a)}
    \includegraphics[width=0.4\textwidth]{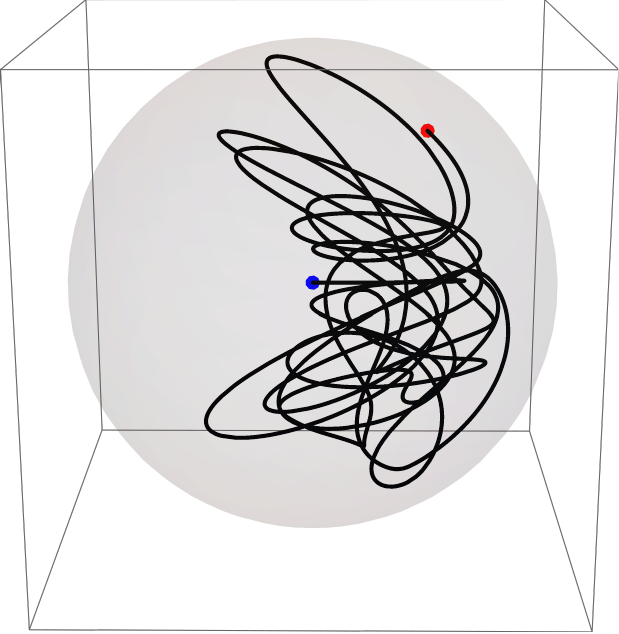} \qquad
    \textbf{b)}
     \includegraphics[width=0.48\textwidth]{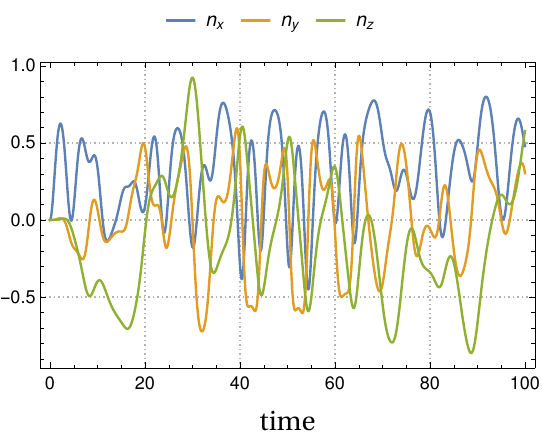}
    \caption{\label{fig:Bloch0}\textbf{An exemplary time-evolution of Bob's reduced density matrix $\rho_B(t)$.} \textbf{a)} The trajectory on the Bloch ball starts at the maximally mixed state (blue dot) and ends up in a partially mixed state (red dot). \textbf{b)} The plot of the corresponding components of the Bloch vector.}
\end{figure}

In order to compare two time-evolutions of Bob's reduced state,  $\rho_B(t)$ and $\rho_B'(t)$, we shall compute the Euclidean distance in the Bloch ball between the two corresponding trajectories
    \begin{equation}
   D(t) = \sqrt{ \left[n_x(t) - n_x'(t)\right]^2 + \left[n_y(t) - n_y'(t)\right]^2 + \left[n_z(t) - n_z'(t)\right]^2},
   \label{lyapunov-distance}
\end{equation}
and plot the quantity $\log ( D(t)/D(0) )$. Such an approach to detect the chaos in quantum systems was also adopted in \cite{QChaos94,QChaos20}. Fig. \ref{fig:exponent} shows a comparison of the logarithmic distance between the trajectories for exemplary dynamics \eqref{GP_qubits} with parameters $a_1 = a_2 = b_2 = d =1$, $b_1 = c = 2$. For the initial states we have chosen the maximally entangled state $\ket{\psi(0)} = (\ket{00}+\ket{11})/\sqrt{2}$ and its perturbation \eqref{psi0} with $x = 0.00005$, which gives $D(0) = 0.0001$. The plots are characteristic to chaotic dynamics and suggest a positive Lyapunov exponent $\lambda>0$ --- see Appendix \ref{app:lyapunov}. Furthermore, the value of the Lyapunov exponent is larger for dynamics with a larger value of the nonlinear coupling, as expected.

\begin{figure}[h]
\includegraphics[width=0.7\textwidth]{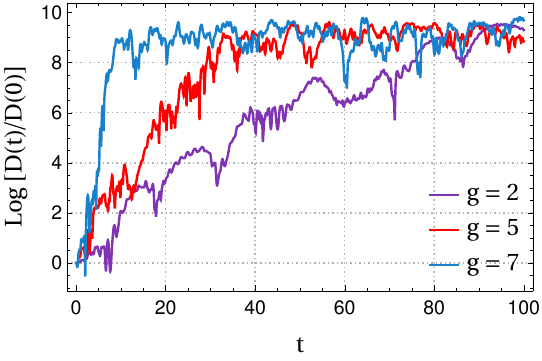}
        \caption{\textbf{Evidence for chaos in nonlinear two-qubit dynamics.} See text for the description.}
    \label{fig:exponent}
\end{figure}

\section{\label{sec:sign}Superluminal signalling}

In the previous section we have analysed several general properties of nonlinear dynamics in bipartite quantum systems. The latter exhibits certain distinctive features, such that the increase of entanglement and chaotic behaviour of trajectories. Up to this point, we have studied the impact on dynamics of the choice of the global initial state, which has been prepared in the common causal past of Alice and Bob (recall Fig. \ref{fig:Gisin}). We are now interested in the possibility of effectuating \emph{operational} superluminal signalling from Alice to Bob. In order to do so, we assume that Alice can encode a bit of information in the system through some local operation on her particle. Then, we inspect whether Bob can infer this bit from the local detection statistics he would obtain by measuring his particle. If this is the case, then superluminal signalling is (statistically)  possible, because we can always arrange the spacetime setup in a way so that the readout event lies outside of the causal future of the sending event. As we shall see, for the general case of local dynamics (i.e., for generic values of parameters of the local Hamiltonians $H_A$ and $H_B$), the no-signaling principle (3a) is violated. This characteristic of dynamics \eqref{GP_qubits} will be investigated across three possible communication protocols.

\subsection{\label{sec:basis}Signalling through the choice of the measurement observable}

We first consider the scenario employed in Gisin's argument. Let Alice and Bob share an entangled state $\ket{\psi} \in \h_A \otimes \h_B$.  Suppose now that Alice can measure one of two possible observables, $X, X' \in \B(\h_A)$, which do not commute. For simplicity, and without loss of generality, we assume that both observables are one-dimensional projectors, $X = \ket{\varphi}\bra{\varphi}$, $X' = \ket{\varphi'}\bra{\varphi'}$, with $0< |\braket{\varphi}{\varphi'}| < 1$. Now, because the dynamics \eqref{G-P-system2} is defined only for pure states in $\h_A \otimes \h_B$ we do need to assume \emph{some} state-update rule after the measurement. We adopt the standard von Neumann postulate.

Let $\ket{\psi(0)}$ be the state prepared at the source. When Alice measures the observable $X$ at a moment $t_0 \geq 0$ (see Fig. \ref{fig:Gisin} \textbf{b)}), she will get the outcome 1 with probability $p_1 = \bra{\psi(t_0)} X \otimes \bone \ket{\psi(t_0)}$ and the global state will be projected to $\ket{\xi_1} = \tfrac{1}{\sqrt{p_1}}(X \otimes \bone) \ket{\psi(t_0)}$, and she will get the outcome 0 with probability $p_0 = 1- p_1$, in which case the global state will become $\ket{\xi_0} = \tfrac{1}{\sqrt{p_0}}\big( (\bone - X) \otimes \bone \big) \ket{\psi(t_0)}$, and analogously for $X'$.
Consequently, depending on whether Alice measures $X$ or $X'$, the effective time-evolution of Bob's reduced density operator for $t \geq t_0$ is, respectively,
\begin{align}
    \rho_B(t) = \sum_{x=0,1} p_x \Tr_A \ket{\xi_x(t)}\bra{\xi_x(t)} \quad \text{ or } \quad \rho_B'(t) = \sum_{x=0,1} p_x' \Tr_A \ket{\xi_x'(t)}\bra{\xi_x'(t)},
\end{align}
where $\ket{\xi_x(t)}$ and $\ket{\xi_x'(t)}$ arise from the dynamics \eqref{G-P-system2} with initial states $\ket{\xi_x(t_0)}$ and $\ket{\xi_x'(t_0)}$, respectively. 
Superluminal signalling is possible if Bob can statistically distinguish between $\rho_B(t)$ and $\rho_B'(t)$, at some time-moment $t>t_0$. Note that if we do not assume the local tomography axiom (and hence the no-restriction hypothesis), then signalling may not be operationally possible even if $\rho_B(t) \neq \rho_B'(t)$. Without loss of generality, we can assume that $t_0 = 0$, because Alice's measurement effectively `resets' the initial state for the global dynamics, so that the evolutions $\rho_B(t)$ and $\rho_B'(t)$ do not depend on the source state $\ket{\psi(0)}$.

We performed an extensive numerical study of the two-qubit dynamics determined by Eqs. \eqref{GP_qubits}. As noncommuting observables we have chosen $X = \ket{0}\bra{0}$ and $X' = \ket{+}\bra{+}$, where $\ket{+} = (\ket{0}+\ket{1})/\sqrt{2}$. We found that for generic parameters of local Hamiltonians, with $c,d \neq 0$, and for a generic entangled source state, Bob can always statistically distinguish between $\rho_B(t)$ and $\rho_B'(t)$, at any time $t>0$. This is true, even if Bob can only measure \emph{one} observable, $\ket{0}\bra{0}$. The case with $c = 0$ and/or $d = 0$ is different and will be described in Section \ref{sec:no}. 
Fig. \ref{fig:measurement12} presents two examples of different dynamics of Bob's reduced state, $\rho_B(t)$ and $\rho_B'(t)$.

\begin{figure}[h]
     \centering
         \begin{subfigure}[b]{0.45\textwidth}
         \centering
         \includegraphics[width=0.75\textwidth]{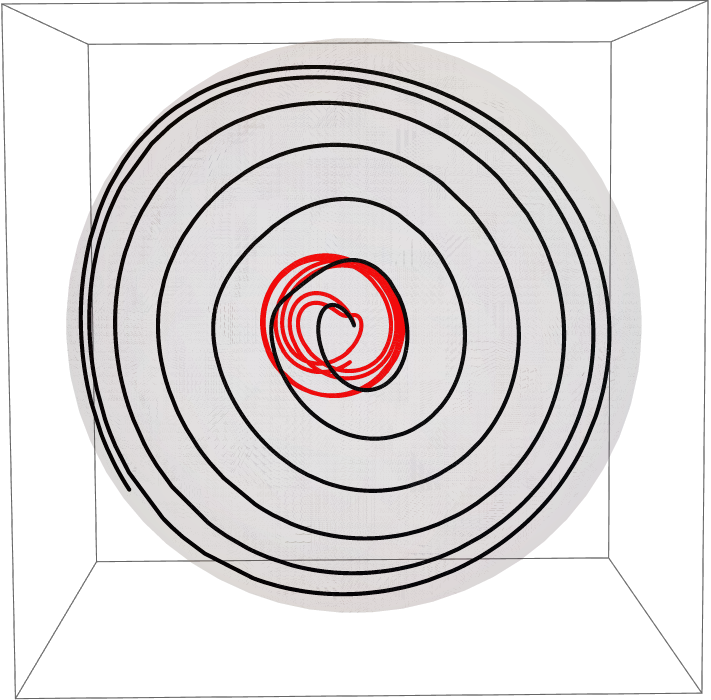}
         \caption{Source state: $\ket{\psi}=(\ket{00}+\ket{11})/\sqrt{2}$}
     \end{subfigure}
     \begin{subfigure}[b]{0.45\textwidth}
         \centering
         \includegraphics[width=0.75\textwidth]{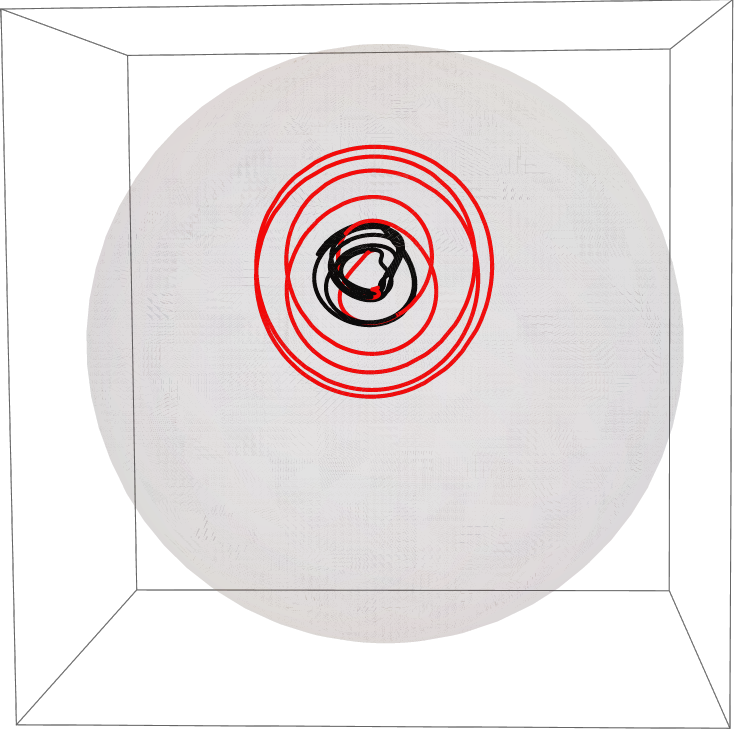}
         \caption{Source state: $\ket{\psi}=(\ket{00}+\ket{01}+\ket{10})/\sqrt{3}$.}
     \end{subfigure}
     \hfill
        \caption{\textbf{Evolution of Bob's state after Alice's measurement of two different observables}. All plots: $a_1=a_2=b_2=c=1$, $b_1=d=2$, $g=3$. Black plot: time-evolution of Bob's state after Alice's measurement of $\ket{0}\bra{0}$; red plot: time-evolution of Bob's state after Alice's measurement of $\ket{+}\bra{+}$.\label{fig:measurement12}}
\end{figure}

What is more, our analysis suggests that the dynamics of Bob's state is very sensitive to the choice of Alice's observables. 
Concretely, we assumed that Alice can choose between $X = \ket{0}\bra{0}$ and $X'= \ket{\varphi_\varepsilon}\bra{\varphi_\varepsilon}$, where $\ket{\varphi_\varepsilon} = \sqrt{1-\varepsilon^2} \ket{0} + \varepsilon \ket{1}$ for a small $\varepsilon > 0$. Such a choice of measurement induces two ensembles of initial states for the dynamics, which differ only very slightly. As it turns out, the nonlinear dynamics \eqref{GP_qubits} drives these ensembles away from each other exponentially fast. In Fig. \ref{fig:measurement22} we plot the logarithm of the distance \eqref{lyapunov-distance} between the trajectories. Because initially we have $\rho_B(0) = \rho'_B(0)$ and hence $D(0)=0$, we have taken $D_\varepsilon(t) = D(t) + \varepsilon$.

\begin{SCfigure}[0.5][h]
        \caption{\textbf{Chaos in the evolution of Bob's effective state induced by Alice's choice of the observable}. The Alice's observables are $X$ and $X'$, as described in the text, with $\epsilon=0.001$. The parameters of dynamics are $a_1=0.1$, $a_2=b_1=b_2=0$, $c=0.3$, $d=0.2$, $g=1$ and the shared entangled state is $\ket{\psi}=(\ket{00}+\ket{11})/\sqrt{2}$.}
         \includegraphics[width=0.55\textwidth]{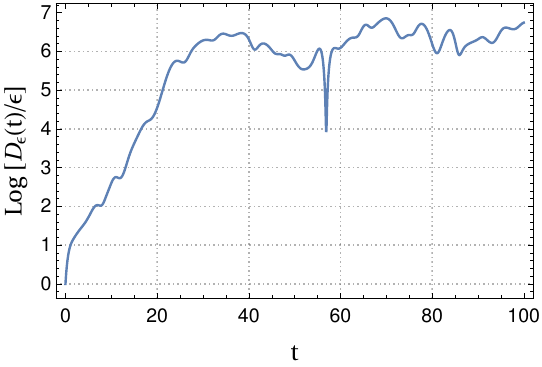}
        \label{fig:measurement22}
\end{SCfigure}

\subsection{\label{sec:meas}Signalling through a local measurement or the lack of it}

We have seen in the previous section that if Alice can choose between measuring two noncommuting observables then she can exploit the dynamics \eqref{GP_qubits} with generic parameters to effectuate superluminal signalling. Suppose now that Alice can only measure \emph{one} observable, say $X= \ket{0}\bra{0}$. Such an assumption constitutes an extreme violation of the no-restriction hypothesis. We assume, however, that Alice can freely decide whether  to perform the measurement on her particle or not, at some moment $t_0 \geq 0$. In this way, she can encode a bit of information into the nonlocal system shared with Bob. The question is whether Bob, being spacelike separated, can (statistically) read out this information. This translates into the problem whether Bob can operationally distinguish at some time $t>t_0$ between two density operators
\begin{align}\label{no_meas}
    \rho_B(t) = \Tr_A \ket{\psi(t)}\bra{\psi(t)} \quad \text{ and } \quad \rho_B'(t) = \sum_{x=0,1} p_x \Tr_A \ket{\xi_x(t)}\bra{\xi_x(t)},
\end{align}
where $\ket{\psi(0)}$ is the state prepared at the source and $\ket{\xi_x(t_0)}$ is defined as in the previous subsection. As previously, we can set $t_0 = 0$ without loss of generality, because assuming that $t_0 > 0$ is equivalent to assuming that the source state was $\ket{\psi(t_0)}$, rather than $\ket{\psi(0)}$ .

Let us note, that this protocol could not have been conceived within the frame of Gisin's theorem without explicitly specifying how does Bob's local device act on a part of an entangled system.

For such a scenario we have performed an extensive numerical study, as in the previous case. The conclusion is again that for generic parameters of the dynamics \eqref{GP_qubits} Bob can always distinguish between the two states \eqref{no_meas}, even if he can only measure one and \emph{the same} observable $X$. 

An example of two time-evolutions of Bob's local state, with and without Alice's projective measurement, is presented in Fig. \ref{fig:measurement}. Observe the difference of evolutions of the $z$ component of the Bloch vector. It means, in particular, that $\Tr X\rho_B(t) \neq \Tr X\rho_B'(t) $ and hence Bob can distinguish the two cases just by measuring the projector $X = \ket{0}\bra{0}$.

\begin{figure}[!htb]
     \centering
     \begin{subfigure}[b]{0.45\textwidth}
         \centering
         \includegraphics[width=0.75\textwidth]{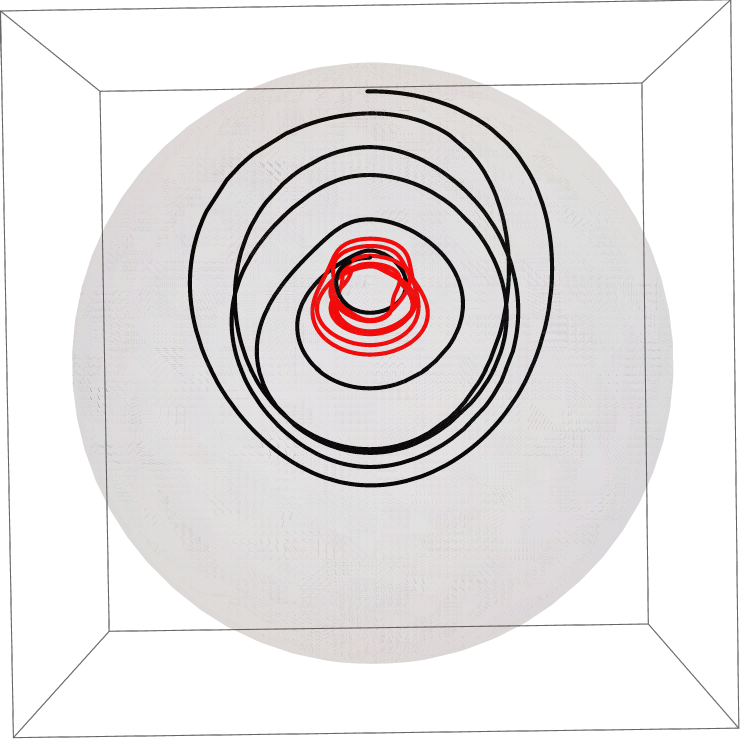}
         \caption{Trajectory in the Bloch ball.}
         \label{fig:initial-conditions-sphere}
     \end{subfigure}
     \hfill
     \begin{subfigure}[b]{0.45\textwidth}
         \centering
         \includegraphics[width=0.9\textwidth]{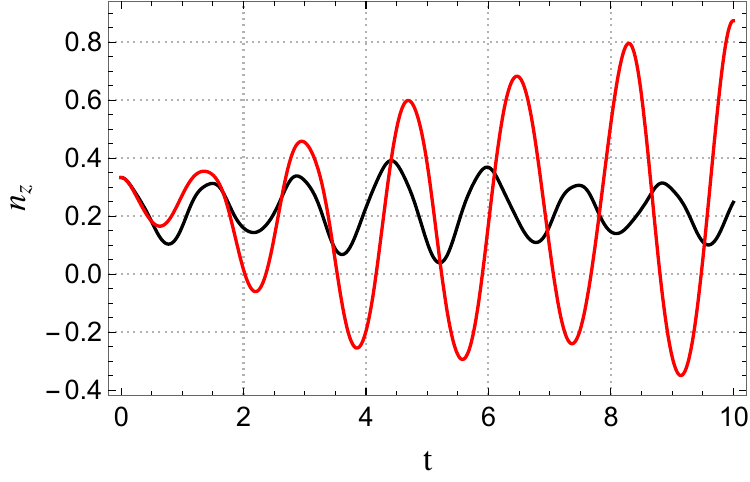}
         \caption{Evolution of the $n_z$ component of the Bloch vector.}
         \label{fig:initial-conditions-norm}
     \end{subfigure}
        \caption{\textbf{Evolution of Bob's state with and without Alice's measurement}.
        All plots: $a_1, a_2, b_2=1$, $b_1, c, d=2$, $g=3$ for the source state $\ket{\psi}=(\ket{00}+\ket{01}+\ket{10})/\sqrt{3}$. Black plot: time-evolution of Bob's state without Alice's measurement; red plot: time-evolution of Bob's state after Alice's measurement.
       }
        \label{fig:measurement}
\end{figure}

\subsection{\label{sec:inter}Signalling through a local intervention}

In the previous subsections we have seen that nonlinear quantum dynamics of the type \eqref{GP_qubits} together with the von Neumann collapse postulate leads, in general, to the violation of the no-signalling principle (3a). This happens, even if we abandon the assumptions (4) and (5) allowing both Alice and Bob to measure just one specific observable.

We shall now consider a scenario, in which there are no projective measurements, and hence the collapse postulate is not needed. Such an approach is motivated by the programme of explaining the measurement problem in terms of decoherence \cite{MeasurementRMP}. It was also adopted in the nonlinear quantum model presented in \cite{Kaplan2022}.

Concretely, we assume that the measurement process is actually a complex local interaction of the quantum system at hand with the (quantum) measuring device. In order to make this description compatible with our model of nonlocal nonlinear dynamics one would need to specify an equation of type \eqref{eq-gen} on the Hilbert space $\h_A \otimes \h_{A'} \otimes \h_B \otimes \h_{B'}$ involving the nonlocal system $AB$ and the local environments $A', B'$. Such an equation induces an inherently nonlocal dynamics, but otherwise concords with the viewpoint underlying decoherence theory, assuming that at the fundamental level all physical systems ought to be described by pure, though entangled, states.

In order to inspect the status of the no-signalling principle we consider a simple protocol, which bypasses the problem of modelling the interactions between the studied systems and the measuring devices. Namely, suppose that Alice can freely change  the local Hamiltonian, $H_A$. We shall call such an operation a \emph{local intervention}. Physically, it could be realised e.g. by a tunable local magnetic field interacting with Alice's particle. 

Let $\ket{\psi(0)}$ be the state prepared at the source and let us denote by $\ket{\psi(t)}$ or $\ket{\psi'(t)}$  its time-evolution under dynamics \eqref{G-P-system2} with $H_A$ or $H_A'$, respectively. Alice can use her freedom to choose between $H_A$ and $H_A'$ to encode a bit of information in the system shared by Bob. Again, without loss of generality, we can assume that she does so a time $t=0$. Then, the issue of superluminal signalling amounts to the question whether Bob can operationally distinguish at time $t > 0$ between 
\begin{align}\label{interv}
    \rho_B(t) = \Tr_A \ket{\psi(t)}\bra{\psi(t)} \quad \text{ and } \quad \rho_B'(t) = \Tr_A \ket{\psi'(t)}\bra{\psi'(t)}.
\end{align}

In the case of two-qubit dynamics \eqref{GP_qubits} Alice could change either one of the diagonal parameters of the Hamiltonian, $a_1, a_2 \in \R$, or the off-diagonal one $c \in \mathbb{C}$. Because of the symmetry \eqref{gauge} it is sufficient to consider just one of the diagonal parameters, say $a_1$. The outcome of our numerical studies is that Alice can always signal to Bob, by manipulating the parameter $a_1$, or $|c|$, or $\arg(c)$.
For generic parameters with $c,d \neq 0$, Bob could see the difference even if he can only measure a single observable $\ket{0}\bra{0}$ on his particle. This happens even if the source state is separable because the dynamics \eqref{GP_qubits} is entangling.

We illustrate the phenomenon of signalling through a local intervention in Figs. \ref{fig:signalling} and \ref{fig:c}. Firstly, we take three source states with different degrees of entanglement: a maximally entangled Bell state in Fig. \ref{fig:signalling-Bell}, an entangled but not maximally entangled state in Fig. \ref{fig:signalling-entangled}, and a separable state in Fig. \ref{fig:signalling-separable}. The red and blue plots illustrate the evolution of Bob's effective state for Alice's choice \( a_1 = 0.3 \) and \( a_1 = 0.2 \), respectively.
\begin{figure}[t]
\hfill
\begin{subfigure}[t]{0.3\textwidth}
    \includegraphics[width=0.7\linewidth]{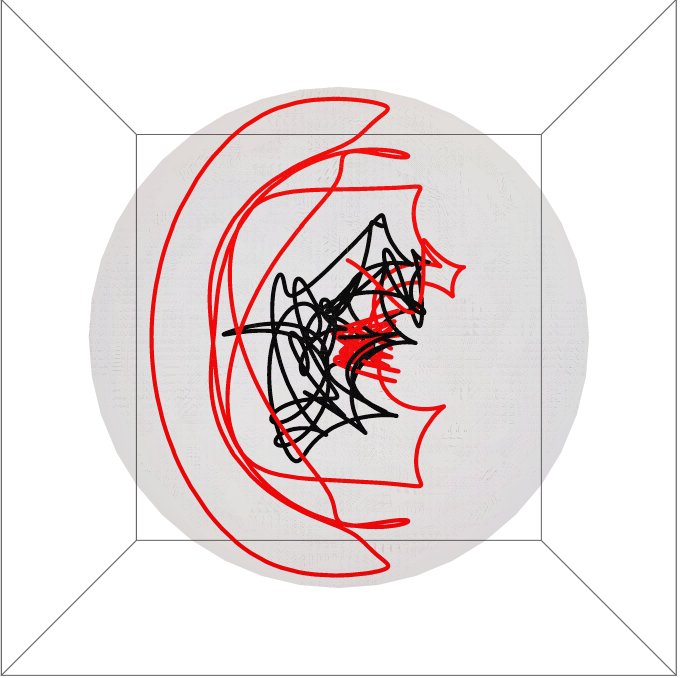}    \includegraphics[width=0.9\linewidth]{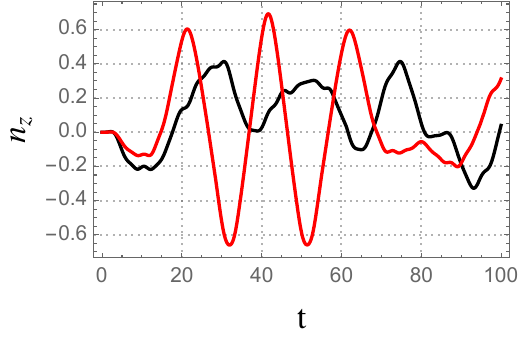}
\caption{Time evolution of a maximally entangled state (Bell state) $\ket{\psi(0)}=(\ket{00}+\ket{11})/\sqrt{2}$.}
\label{fig:signalling-Bell}
\end{subfigure}
\hfill
\begin{subfigure}[t]{0.3\textwidth}
    \includegraphics[width=0.7\linewidth]{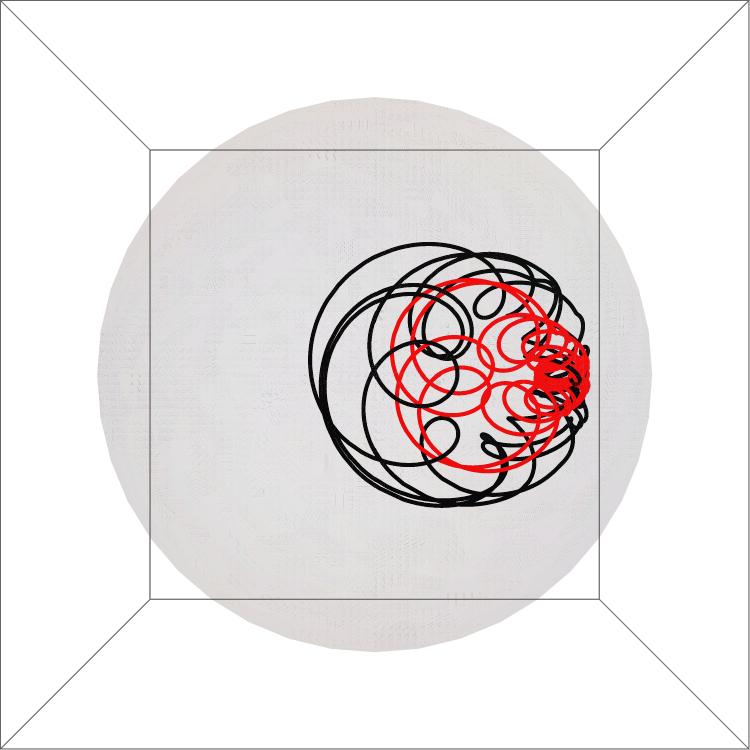}
    \hfill
  \includegraphics[width=0.9\linewidth]{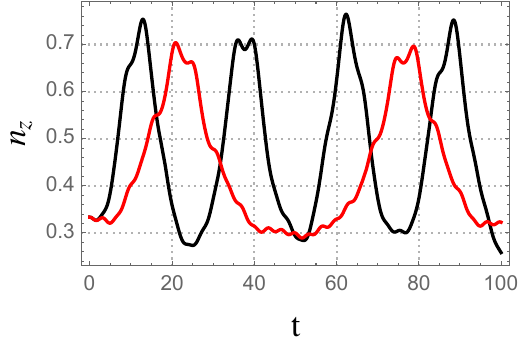}\hfill
\caption{Time evolution of an entangled, but not maximally entangled, state $\ket{\psi(0)}=(\ket{00}+\ket{01}+\ket{10})/\sqrt{3}$.}
\label{fig:signalling-entangled}
\end{subfigure}
\hfill
\begin{subfigure}[t]{0.3\textwidth}
    \includegraphics[width=0.7\linewidth]{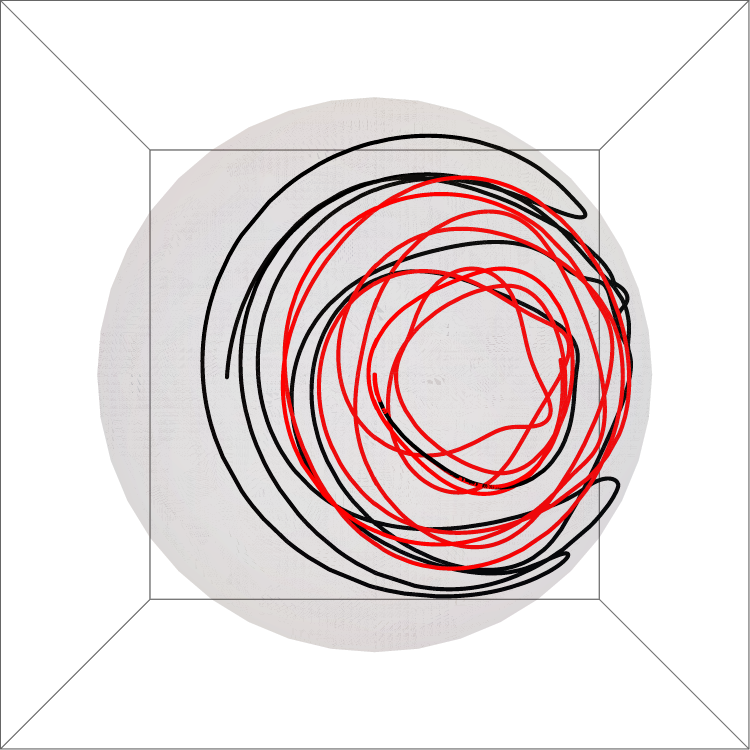}\hfill
  \includegraphics[width=0.9\linewidth]{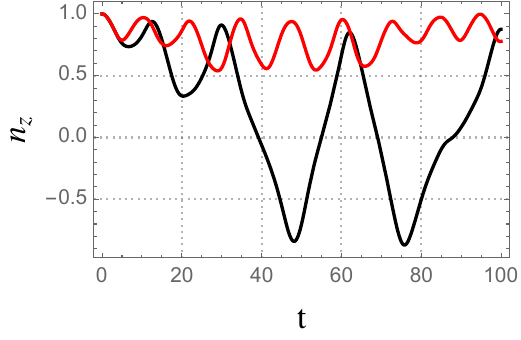}\hfill
\caption{Time evolution of a separable state $\ket{\psi(0)}=\ket{00}$.}
\label{fig:signalling-separable}
\end{subfigure}
\caption{\textbf{Signalling through a local intervention}. All plots: $a_2,b_1,b_2,d=0.1$, $c=0.4$, $g=1$. Red plot: $a_1=0.3$, black plot: $a_1=0.2$. First row: evolution of a Bloch vector in the Bloch ball. Second row: evolution of the $z$-component of Bloch vector in function of time.
}
\label{fig:signalling}
\end{figure}

Furthermore, the dynamics \eqref{GP_qubits} turns out to be exponentially sensitive not only to the initial state, but also to the parameters of local Hamiltonians. We illustrate this in Fig. \ref{fig:c}. In Figs. \ref{fig:c1} and \ref{fig:c2} we compare the trajectories in the Bloch ball starting with the same initial state, but driven by dynamics with a slightly different parameter $c$. While initially close, the trajectories are eventually driven apart. In order to quantify this effect we extend the phase space  from the Bloch ball $B_3$ to a cylinder $B_3 \times [c_1,c_2]$ and equip it with the Euclidean distance
    \begin{equation}
   D(t) = \sqrt{ \left[n_x(t) - n_x'(t)\right]^2 + \left[n_y(t) - n_y'(t)\right]^2 + \left[n_z(t) - n_z'(t)\right]^2 + \left[c - c'\right]^2}.
   \label{lyapunov-distance-c}
\end{equation}
This allows us to detect chaos in dynamics \eqref{GP_qubits} with respect to a change in the parameters of $H_A$. The plot in Fig. \ref{fig:c3} provides evidence that the corresponding Lyapunov exponent is positive. On the operational side, it implies that even a minute change in Alice's local Hamiltonian eventually results in a sizable effect on Bob's. Consequently, the would-be superluminal communication protocol would be very efficient.

\begin{figure}[h]
\begin{subfigure}[t]{0.25\textwidth}
    \includegraphics[width=1.1\linewidth]{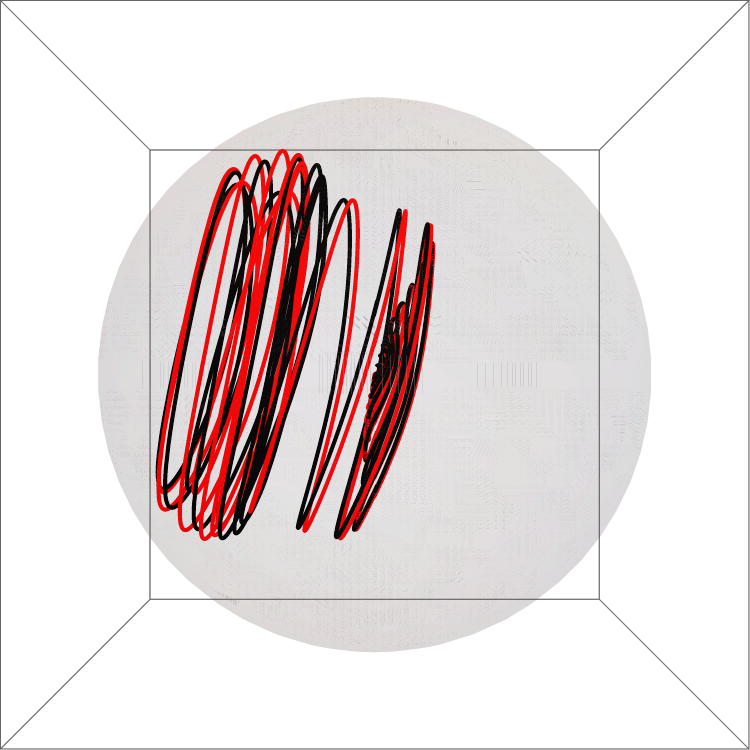}
    \caption{Time evolution over $t\in [0,30]$.}
    \label{fig:c1}
\end{subfigure}
\hfill
\begin{subfigure}[t]{0.25\textwidth}
    \includegraphics[width=1.1\linewidth]{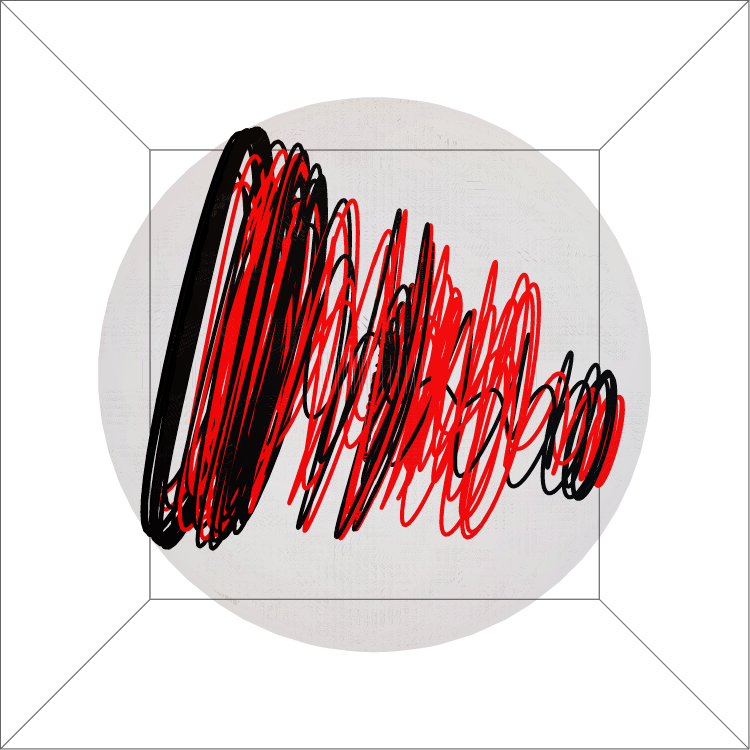}
    \caption{Time evolution over $t\in [0,100]$.}
    \label{fig:c2}
\end{subfigure}
\hfill
\begin{subfigure}[t]{0.4\textwidth}
    \includegraphics[width=0.9\linewidth]{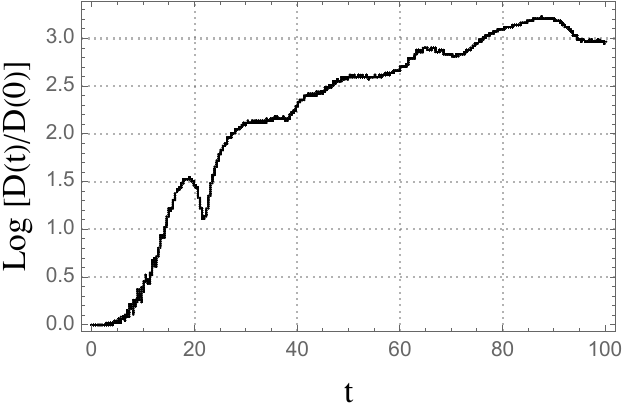}\hfill
    \caption{Logarithm of the Euclidean distance between the trajectories on a cylinder.}
    \label{fig:c3}
\end{subfigure}
\caption{\textbf{Illustration of exponential sensitivity of dynamics \eqref{GP_qubits} on parameters of $H_A$.} Both plots: $a_1,b_1,g=1$, $a_2=0.5$, and $b_2,d=2$. Black plot: $c=2$; red plot: $c=2 + \epsilon$, with $\epsilon = 0.001$. The source state was the Bell state $\ket{\psi(0)}=(\ket{00} + \ket{11})/\sqrt{2}$.}\label{fig:c}
\end{figure}

\section{\label{sec:no}Nonsignalling nonlinear nonlocal dynamics}
In the previous section we have seen that for generic parameters the dynamics \eqref{GP_qubits} facilitates operational superluminal signalling, even without the collapse postulate. It strongly suggests that general nonlinear dynamics \eqref{G-P-system2} suffers from this problem.  We now show, however, that there exists a special class of dynamics \eqref{G-P-system2}, which are consistent with the no-signalling principle.

The nonlinearity, which we consider in this paper \eqref{general-nonlinearity} distinguishes a specific basis of the Hilbert space $\h$. From formula \eqref{G-P-system2} it is clear that the mixing between different components of the state vector, expanded in this basis, is induced by the local Hamiltonians. Consequently, Proposition \ref{prop:sol} shows that if both $H_A$ and $H_B$ are diagonal in the basis distinguished by nonlinearity, then Eqs.  \eqref{G-P-system2} decouple and admit an analytic solution. This leads to the following result.

\begin{prop}\label{prop:NS_i} 
    If both local Hamiltonians $H_A$, $H_B$ are diagonal in the basis singled out by the nonlinearity, then Bob's reduced density operator \eqref{eq:partial-trace} does not depend on the parameters of Alice's local hamiltonian $H_A$, and vice versa.
\end{prop}
\begin{proof}
    Recall from Proposition \ref{prop:sol} that for diagonal Hamiltonians $H_A$ and $H_B$ the system of equations \eqref{G-P-system2} has a unique analytic solution of the form \eqref{sol}. We can thus compute analytically Bob's local state \eqref{eq:partial-trace}. Concretely, 
\begin{equation}\label{rhoB_diag}
    \bra{k} \rho_B(t) \ket{k'} = \sum_{j} \alpha_{jk}(t) \overline{\alpha_{jk'}(t)} = \sum_{j} \aleph^0_{jk} \aleph^0_{jk'} e^{-i \big[ b_k-b_{k'}+f_{jk} (\aleph^0_{jk})-f_{jk'}(\aleph^0_{jk'}) \big]t + i \big( \varphi^0_{jk} - \varphi^0_{jk'} \big)},
\end{equation}
where the notation is as in \eqref{X} and \eqref{sol}. We thus see that $\rho_B(t)$ does not depend on $a_j$'s. Analogously, one can show that $\rho_A(t)$ does not depend on $b_k$'s.
\end{proof}
We thus see that if both $H_A$ and $H_B$ are diagonal in the preferred basis, then the dynamics \eqref{G-P-system2} cannot facilitate superluminal signalling through local interventions. Let us note that even though the linear operators $H_A \otimes \bone$, $\bone \otimes H_B$ and the nonlinear operator $K$ defined in \eqref{general-nonlinearity} have a common basis, they \emph{do not} commute. Indeed, with $\ket{\psi} = \sum_{j,k} \alpha_{jk} \ket{jk}$ we have
\begin{align}
 & (H_A \otimes \bone) K(\ket{\psi}) = \sum_{j,k} f_{jk} ( |\alpha_{jk}| ) \alpha_{jk} a_j \ket{j,k}  \neq  K\big( (H_A \otimes \bone) \psi \big) =  \sum_{j,k} f_{jk} ( |\alpha_{jk} a_j| ) \alpha_{jk} a_j \ket{j,k}.
\end{align}

On the other hand, if we assume the (von Neumann) collapse postulate, then Alice can, in general, signal to Bob (and vice versa) through a local measurement. However, if we restrict the possible observables that Alice and Bob can choose to measure, then this conclusion can be avoided.

\begin{prop}\label{prop:NS_m} 
    Assume that the local Hamiltonian $H_B$ is diagonal in the basis singled out by the nonlinearity and let $M\in \B(\h_A)$ be an observable. Then, the diagonal elements of Bob's local density matrix, $\bra{k} \rho_{B}(t) \ket{k}$ are constant and do not depend on whether $M$ was measured at some time $t_0 \leq t$ or not. Symmetrically, the same holds with $A$ and $B$ exchanged.
\end{prop}
\begin{proof}
This follows immediately from Proposition \ref{prop:diag}. Indeed, let $\ket{\psi(0)}$ be the state prepared at the source and let us consider the spectral decomposition of the measured observable, $M = \sum_r m_r \P^r$. According the von Neumann collapse postulate, immediately after the measurement at time $t_0 \geq 0$ the state of the system is projected to $\ket{\xi^r(t_0)} = \tfrac{1}{\sqrt{p_r}} (\P^r \otimes \bone) \ket{\psi(t_0)}$ with probability $p_r = \bra{\psi(t_0)} (\P^r \otimes \bone) \ket{\psi(t_0)}$. Because the standard projection postulate is consistent with the no-signalling principle, Bob's density matrix does not feel Alice's measurement, 
\begin{align}
    \rho_B'(t_0) = \sum_r p_r \rho_B^r(t_0) = \sum_r p_r \Tr_A \ket{\xi^r(t_0)} \bra{\xi^r(t_0)} = \Tr_A \ket{\psi(t_0)} \bra{\psi(t_0)} = \rho_B(t_0).
\end{align}
Now, Proposition \ref{prop:diag} implies that for any $t \geq t_0$, 
\begin{equation}
\begin{aligned}
    \bra{k} \rho_B'(t) \ket{k} & = \sum_r p_r \bra{k} \rho_B^r(t) \ket{k} = \sum_r p_r \bra{k} \rho_B^r(t_0) \ket{k} \\
    & = \bra{k} \rho_B'(t_0) \ket{k} = \bra{k} \rho_B(t_0) \ket{k} = \bra{k} \rho_B(t) \ket{k}.
\end{aligned}
\end{equation}
The same reasoning applies with $A$ and $B$ exchanged.
\end{proof}

Proposition \eqref{prop:NS_m} shows that if $H_B$ is diagonal in the preferred basis and Bob can only measure observables of the form
\begin{align}
 M_B = \sum_k m_k^B \ket{k} \bra{k}, 
\end{align}
which commute with the preferred basis $\{ \ket{k} \}$ of $\h_B$, then Alice cannot signal to him neither using local interventions nor projective measurements. 

Let us stress that if Bob would be able to measure an observable, which is not aligned with the preferred basis, then he could statistically infer Alice's manipulations, whether she acted through local interventions or projective measurements. In Fig. \ref{fig:dzero-states} we illustrate this fact for the 2-qubit system evolving according to Eqs. \eqref{GP_qubits} with $d=0$. 

\begin{SCfigure}[1.5][h]
        \caption{\textbf{Evolution of $\rho_B(t)$ for a diagonal Hamiltonian $H_B$.} Both plots: $a_2=1$, $b_1=2$, $b_2=1$, $c=2$, $d=0$, $g=3$ for an initial state $\ket{\psi(0)}=(\ket{00}+\ket{11})/\sqrt{2}$.
        Red plot: $a_1=2$; black plot: $a_1=1$. In accordance with Proposition \ref{prop:diag} the trajectory of $\rho_B(t)$ is constrained to the $z=0$ plane (yellow disk).}
         \includegraphics[width=0.25\textwidth]{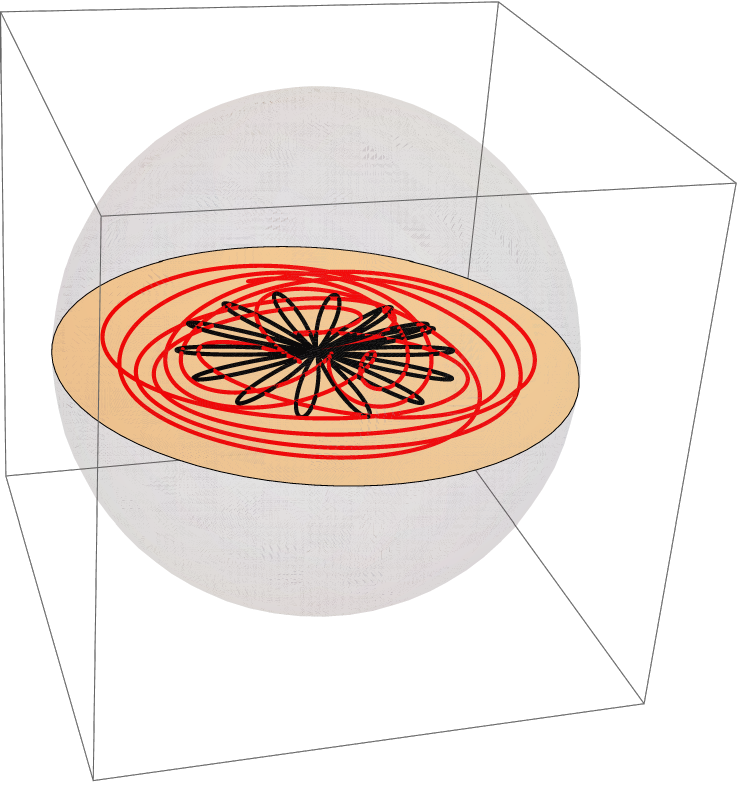}
         \label{fig:initial-conditions-sphere}
        \label{fig:dzero-states}
\end{SCfigure}

In summary, we have shown that --- with the assumptions (1a--1d) and (2) --- the nonlocal nonlinear dynamics of the form \eqref{G-P-system2} may be consistent with the no-signalling principle (3a) only if at least one of the local Hamiltonians is diagonal in the basis singled out by nonlinearity. Even then, the demand of no-signalling requires an extreme violation of the no-restriction hypothesis: We have to assume that the only operationally measurable observables are the ones, which commute with the preferred basis.

\section{\label{sec:conc}Conclusions}

Gisin's no-go theorem \cite{Gisin1989} was established 35 years ago and it is widely considered to undermine deterministic nonlinear quantum dynamics as a valid physical model (see e.g. \cite{OR}). However, this result is founded on some prerequisite assumptions, which could and have been questioned in several models \cite{Czachor1998,CzachorMarcin,Czachor02,Kent05,Helou17,Caban20,Caban21,Kaplan2022,Paterek24}. It is thus of high importance to understand what are the indispensable assumptions underlying the argument and what are its limits of applicability.

In this work we scrutinised Gisin's theorem in its version presented with Simon and Bu\v{z}ek in \cite{Gisin2001}. We conclude that it is based on four main premises, as discussed in detail in Section \ref{sec:2}. We unveiled all possible `loopholes', that one can exploit to construct models of deterministic nonlinear quantum dynamics, which may be consistent with the no-signalling principle.

In particular, we observed that the argument uses the no-restriction hypothesis, so one could explore the possibility of nonlinear dynamics of states within the frame of General Probabilistic Theories \cite{GPT_review}. One of the interesting options would be to consider models of dynamics in the context  of the so-called \emph{complete extension} \cite{CEP}: A state $s_{AE}$ is an extension of a state $s_A$ if its reduced state (as defined in the considered theory) is $s_A$. It is  called a complete extension, if any other extension can be generated from it by a local action on $E$ --- in analogy to the action of a local unitary on a global pure quantum state. Complete extensions can be seen as minimal variants of purifications in quantum theory and a General Probabilistic Theory satisfies the Complete Extension Postulate (CEP) if it only allows for complete extensions. Now, it seems possible that theories  which do not satisfy CEP can admit interesting nonlinear dynamics which are no-signaling and have nontrivial extensions, in some analogy to the Steinspring theorem in quantum mechanics. Another option would be to search for theories with CEP which, however, do not ``hyperdecohere'' to quantum mechanics, unlike quantum mechanics which decoheres to classical mechanics. The existence of such theories is not a priori excluded \cite{CEP}.

Another unveiled loophole concerns the assumed locality of the dynamics. Concretely, the theorem in \cite{Gisin2001} says that, given a Hilbert space $\h = \otimes_n \h_n$, any map $\E: \h_n \to S(\h_n)$ must be linear and completely positive, but it does not say anything about the possible maps $\overline{\E}: S(\h_n) \to S(\h_n)$,  $\widetilde{\E}: \h \to \h$ or $\widetilde{\E}’: \h \to S(\h)$. In  \cite{Caban20,Caban21} and \cite{Paterek24} the authors presented explicit examples of nonlinear, yet non-signalling, dynamical maps of the form $\overline{\E}$. In this work we studied some simple class of nonlocal maps of the type $\widetilde{\E}: \h_{A} \otimes \h_B \to \h_{A} \otimes \h_B$. These could be seen as a toy model for more elaborate physical models, such as the one presented in \cite{Kaplan2022}. The latter model involves nonlinearities resulting from expectation values of quantum-field-theoretic operators, hence it is inherently nonlocal, and singles out the `position' as a preferred observable.

A general conclusion from our studies is that, if we keep the static structure of quantum mechanics (i.e. assumptions (1a--2)), then in order to make nonlinear dynamics of the form \eqref{eq-gen2} consistent with the no-signalling principle one needs to give up the no-restriction hypothesis anyway. It is important to stress that the transfer of information between the spacelike separated parties does not necessarily require projective measurements, as we showed in Sec. \ref{sec:inter}. It means that even in theories, which do not adopt the collapse postulate, like e.g. \cite{Kaplan2022}, one needs to carefully check the impact of local interventions, through e.g. auxiliary fields, on the dynamic of global quantum states.

The main lesson from Gisin's theorem, supported by this work, is that \cite{Gisin2001}: ``It is clearly difficult to modify just parts of the whole structure'' of quantum mechanics. The major problem in nonlinear quantum dynamics is how to extend a nonlinear map from a single system to composite systems in a consistent way. This problem concerns also the class of dynamical maps \eqref{eq-gen2} considered here. Indeed, it is unclear how to extend the dynamics \eqref{eq-gen2} to a map on $\h_A \otimes \h_{A'} \otimes \h_B \otimes \h_{B'}$, so that it involves the agents' local environments. One could, however, argue in favour of the general form of global dynamics \eqref{eq-gen}, which determines the effective dynamics of the subsystems' states --- somewhat in the spirit of \cite{Kaplan2022}.

It seems that a categorical rejection of physical models based on deterministic nonlinear quantum dynamics might be premature. Indeed, recently there is a growing interest in such models, which can yield concrete testable predictions \cite{Caban23,KaplanPRL1,KaplanPRL2}. Finally, let us also point out that the no-signalling principle should not be treated as an unbreakable rule for constructing physical models, as it is in fact violated in nonrelativistic quantum dynamics \cite{PRA2020}. One can thus envisage \emph{effective} nonrelativistic models involving deterministic nonlinear dynamics of quantum states, as long as one can show that any superluminal signalling effects are irrelevant within the physical regime, in which the model is expected to apply.

\section*{Acknowledgements}
ME would like to thank Karol {\.Z}yczkowski for the enlightening discussions on the detection of chaos.  
This work was supported by the National Science Centre in Poland under the research grant Maestro (2021/42/A/ST2/0035). PH acknowledges support by the Foundation for Polish Science through IRAP project co-financed by EU within Smart Growth Operational Programme (contract no. 2018/MAB/5).

\appendix
\section{Methods. Analysis of Chaotic Dynamics}
\label{app:lyapunov}
In the main part of the text we have argued that the dynamical system described by equations \eqref{GP_qubits} exhibits chaotic behaviour. 
A central tool to detect chaos in dynamical systems are the \emph{Lyapunov exponents}. Concretely, if $D(t)$ is a distance between two trajectories in function of time, then the maximal Lyapunov exponent is defined as
\begin{equation}
    \lambda = \lim_{t \to \infty} \lim_{D(0) \to 0} \frac{1}{t} \log\left(\frac{D(t)}{D(0)}\right).
    \label{def:exponent-lyapunov}
\end{equation}
If $\lambda > 0$ then $D(t) \approx e^{\lambda t} D(0)$, so that two initially close trajectories are driven away from each other exponentially fast. 

In quantum mechanics one can study the distance between trajectories in the space of quantum states \cite{Karol93,ChaosQM}. For a single-qubit effective dynamics, which we studied in Section \ref{sec:sign}, the trajectories live in the Bloch ball and for $D(t)$ we take the Euclidean distance \eqref{lyapunov-distance}. Such an approach provides evidence for the exponential sensitivity of the trajectory on the initial state. In order to study the sensitivity of trajectories on parameters of the dynamics, rather than on the initial state, we need to treat the desired parameter as a variable and consider the distance in the extended space \eqref{lyapunov-distance-c}.

A qualitative study of Lyapunov exponents in the system \eqref{GP_qubits} is beyond the scope of this paper. Nevertheless, we can estimate the maximal Lyapunov exponent from the numerical studies and treat it as a signature of chaos (cf. \cite{ChaosQM94,QChaos20}).

In order to illustrate the technique, we analyse the plot presented in Fig. \ref{fig:exponent}. 
We see that, on the average, the plotted function grows and then flattens out. The approximately linear growth implies an exponential growth of the distance between the two trajectories. The final flattening occurs because the phase-space, that is the Bloch ball, is compact and hence the maximal distance between the trajectories is bounded from above. 

\begin{figure}[!htb]    \includegraphics[width=0.6\textwidth]{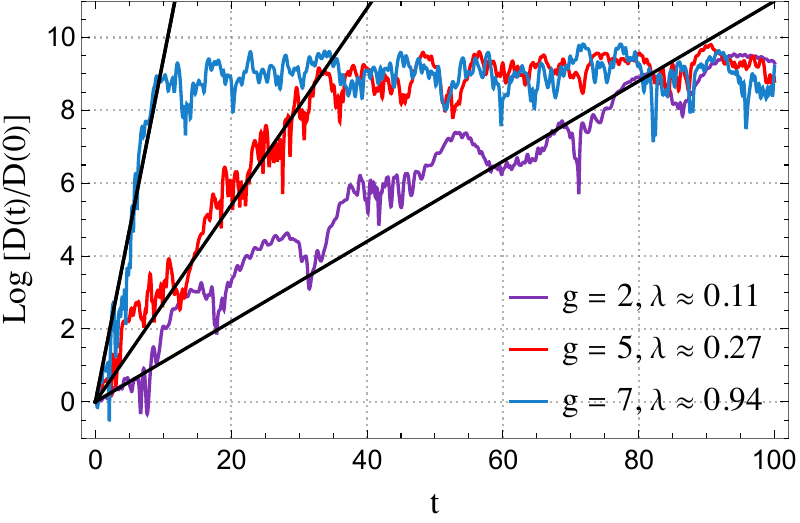}
\caption{\textbf{Estimates of the Lyapunov exponents with a linear regression fit.} Linear regression for plots on Fig. \ref{fig:exponent}.}
    \label{Exponents}
\end{figure}

In order to estimate the maximal Lyapunov exponent we identified a range of $t$ for which the distance between two trajectories -- original one, and varied by a small factor $\varepsilon$, increases exponentially, as $t \in [0,t_{max}]$. Then we used linear regression $\displaystyle y = a t + b$ with $b=0$ to obtain an equation $\displaystyle y = a t$ with the coefficient of variation $CV=\sigma/\mu$, where $\sigma$ is the standard deviation of the dataset and $\mu$ is the mean of the dataset. From this, we estimate the Lyapunov exponent to be $\lambda = a$. 
Further, we noted that small changes in the choice of the range of $t$ for which we used linear regression does not significantly affect values of the Lyapunov exponent and the coefficient of variation -- a change of range of $t$ by $5\%$ results in change of $\lambda$ by $\delta \lambda$ and change in $CV$ by $\delta CV$. The results for Fig. \ref{Exponents} are presented in the Tab. \ref{tab}. 

Thus, given that $\lambda > 0$ and given the value of $CV$, these particular system (i.e. a system with this specific dynamics, with certain parameters varied by a given factor $\epsilon$, which affects the system only in one direction on the Bloch sphere) exhibits chaotic behavior. 

\begingroup
\setlength{\tabcolsep}{10pt} 
\renewcommand{\arraystretch}{1.5} 
\begin{table}[!htb]
\centering
\begin{tabular}{ ||c|c|c|c||}
\hline
$g$ & $t \in [0, t_{max}]$ & $\lambda \pm \delta \lambda$ & $CV$ $\pm \delta CV$ \\
\hline
2 & $t\in [0,95]$ & 0.1096 $\pm$ 0.003 & 0.2119 $\pm$ 0.005 \\
5 & $t\in [0,35]$ & 0.2748 $\pm$ 0.009 & 0.0669 $\pm$ 0.011 \\
7 & $t\in [0,10]$ & 0.9395 $\pm$ 0.032 & 0.1271 $\pm$ 0.032 \\
\hline
\end{tabular}
\caption{\textbf{Lyapunov exponents for different values of $g$} corresponding to Fig. \ref{Exponents}. The system exhibits chaotic behaviour in range $t\in[0,t_{max}]$; $\lambda$ is a Lyapunov exponent and $CV$ is a coefficient of variation. Errors were calculated as average changes in values of $\lambda$ and $CV$ in fits with $t_{max} \pm 5\% \cdot t_{max}$.}
\label{tab}
\end{table}
\endgroup

Generalising this to other systems of a certain type -- for example those, which evolution is perturbated in another direction, or even all systems with a given non-linear potential -- requires more thorough analysis, that accounts for changes in seven free variables. This analysis is beyond the scope of our considerations in this paper. 

\newpage
\bibliography{HAB}{}
\bibliographystyle{plain}

\end{document}